\documentclass{llncs}

\usepackage[T1]{fontenc}
\usepackage{lmodern}
\usepackage[USenglish]{babel}
\usepackage{url}
\usepackage{latexsym}
\usepackage{subfig}
\usepackage{epsfig}
\usepackage{color}
\usepackage[boxed]{algorithm2e}
\usepackage[latin1]{inputenc}
\usepackage{graphicx}
\usepackage{amstext}
\usepackage{amssymb}
\usepackage{amsfonts}
\usepackage{amsmath}
\allowdisplaybreaks
\usepackage{multirow}
\usepackage{bussproofs}
\usepackage{changepage}
\usepackage{centernot}
\usepackage{tikz}
\DeclareCaptionType{copyrightbox}

\renewcommand{\paragraph}[1]{\noindent\textbf{#1}}

\newcommand{\comment}[1]{}
\newcommand{\bfstart}[1]{\smallskip\noindent\textbf{#1}}
\newcommand{\n}[1]{\overline{#1}}

\newcommand{\mL}{\mathcal{L}}

\newcommand{\mF}{\mathcal{F}}
\newcommand{\nin}{\noindent{}}

\newcommand{\lb}{\linebreak{}}

\usetikzlibrary{arrows}
\usetikzlibrary{positioning}

\begin{document}

\title{Interpolation Properties and\\ SAT-based Model
  Checking\thanks{\tiny This material is based upon work funded and
    supported by the Department of Defense under Contract
    No. FA8721-05-C-0003 with Carnegie Mellon University for the
    operation of the Software Engineering Institute, a federally
    funded research and development center.  This material has been
    approved for public release and unlimited
    distribution. DM-0000469.}  }

\author{Arie Gurfinkel\inst{1} \and Simone Fulvio Rollini\inst{2} \and Natasha Sharygina\inst{2}}

\institute{Software Engineering Institute, CMU\\
\email{arie@cmu.edu}
\and
Formal Verification Lab, University of Lugano\\
\email{\{simone.fulvio.rollini, natasha.sharygina\}@usi.ch}}

\maketitle

\begin{abstract}
\vspace*{-0.1in}
  Craig interpolation is a widespread method
in
verification, with important applications such as
  Predicate Abstraction, CounterExample Guided Abstraction Refinement
  and Lazy Abstraction With Interpolants.  Most state-of-the-art
  model checking techniques based on interpolation require {\em
   collections} of interpolants to satisfy
particular properties, to  which we refer as ``collectives'';
they do not hold in general for
  all interpolation systems and have to be established for each
  particular system and verification environment.  Nevertheless, no
  systematic approach exists that correlates the individual
  interpolation systems and compares the necessary collectives.  This
  paper proposes a uniform framework, which encompasses (and
  generalizes) the most common collectives exploited in
  verification.  We use it for a systematic study of the collectives
  and of the
constraints they pose on propositional interpolation systems used in SAT-based model checking.
 \end{abstract}

\section{Introduction}


Craig interpolation is a popular approach in 
verification~\cite{McM03,McM04a} with notable applications 
such as Predicate Abstraction \cite{JM05}, 
CounterExample Guided Abstraction Refinement (CEGAR)~\cite{HJM+04},
and Lazy Abstraction With Interpolants
(LAWI)~\cite{McM06}.
 
Formally, given two formulae $A$ and $B$ such that $A \land B$ is
unsatisfiable, a \emph{Craig interpolant} is a formula $I$ such that
$A$ implies $I$, $I$ is inconsistent with $B$ and $I$ is defined over
the atoms (i.e., propositional variables) 
common to $A$ and $B$. It can be seen as an
over-approximation of $A$ that is still inconsistent with $B$\footnote{
We write
$Itp(A \mid B)$ for an interpolant of $A$ and $B$, and $I_A$ when $B$
is clear from the context.}. In model checking applications, $A$
typically encodes some finite program traces, and $B$ denotes error
locations. In this case, an interpolant $I$ represents a set of
\emph{safe} states that over-approximate the states reachable in $A$.

In most verification tasks, a single interpolant, i.e., a single subdivision
of constraints into two groups $A$ and $B$, is not sufficient. For
example, consider the refinement problem in CEGAR: given a spurious
error trace $\pi = \tau_1, \ldots, \tau_n$, where $\tau_i$ is a
program statement, find a set of formulae $X_0, \ldots, X_n$ such that
$X_0 = \top$, $X_n = \bot$, and for $1 \leq i \leq n$, the Hoare
triples $\{X_{i-1}\}\; \tau_i\; \{X_i\}$ are valid. The sequence
$\{X_i\}$ justifies that the error trace is infeasible and is used to
refine the abstraction. 
The solution 
is a
\emph{sequence} of interpolants $\{I_i\}_{i=1}^{n}$ such that:
  $I_i  = Itp (\tau_1 \ldots \tau_i \mid \tau_{i+1} \ldots \tau_n)$ and  
  $I_{i-1} \land \tau_i \implies I_i$. That is, in addition to requiring that each $I_i$ is an interpolant
between the prefix (statements up to position $i$ in the trace) and
the suffix (statements following position $i$), the sequence
$\{I_i\}$ of interpolants must be inductive: this property 
is known as the
\emph{path interpolation property}~\cite{RSS12}. 

Other properties (e.g.,
simultaneous abstraction, interpolation sequence, path-, symmetric-,
and tree-interpolation) are used in 
existing verification frameworks such as IMPACT~\cite{McM06}, Whale~\cite{AGC12a}, FunFrog~\cite{FSS12} and 
eVolCheck~\cite{SFS12}, which implement instances of 
Predicate Abstraction~\cite{JM06}, Lazy Abstraction with
Interpolation~\cite{McM06}, Interpolation-based Function Summarization~\cite{FSS12} and 
Upgrade Checking~\cite{SFS12}.
These properties, to which we refer as \emph{collectives} since they concern collections 
of interpolants, are not 
satisfied by arbitrary sequences of Craig interpolants and must be
established
for each interpolation algorithm and verification technique.

This paper performs a systematic study of collectives in verification and identifies the particular constraints they pose 
    on propositional interpolation systems used in SAT-based model checking. 
  The SAT-based approach provides bit-precise reasoning which is essential both in software and hardware applications, 
  e.g., when dealing with pointer arithmetic and overflow. 
  To-date, there exist successful tools which perform SAT-based model checking (such as CBMC\footnote{http://www.cprover.org/cbmc} and 
  SATABS\footnote{http://www.cprover.org/satabs}), and which integrate it with interpolation (for example,
  eVolCheck and FunFrog).
  However, there is no a framework which would correlate the existing interpolation systems 
  and compare the various collectives.
  This work addresses the problem and contributes as follows:
  
{\em Contribution 1:} This paper, for the first time, collects, identifies, and
uniformly presents the most common collectives imposed on
interpolation by existing verification approaches (see \S\ref{sec:int_prop}).

In addition to the issues related to a diversity of interpolation properties, 
it is often desirable to have flexibility in
 choosing different algorithms for computing different interpolants in a
sequence  $\{I_i\}$, rather than using a single interpolation algorithm (or \emph{interpolation system}) $Itp_S$,
as assumed in the path interpolation example above.
To guarantee such a flexibility, this paper presents
a framework which generalizes the traditional setting
consisting of a single interpolation system to allow for sequences, or
\emph{families}, of interpolation systems. For example, given a family of
systems $\mF = \{Itp_{S_i}\}_{i=1}^n$, let 
$
I_i = Itp_{S_i} (\tau_1, \ldots \tau_i \mid \tau_{i+1} \ldots \tau_n)
$. 
If the resulting sequence of interpolants $\{I_i\}$ 
satisfies the condition of path interpolation, we say that the family $\mF$
has the path interpolation property. 

Families find practical applicability
in several contexts\footnote{The notion of families is additionally a useful technical tool to  make the discussion  and the results 
more general and easier to
compare with the prior work of CAV'12 \cite{RSS12} (which formally defined families
for the first time).}. One
example is LAWI-style verification, where it is desirable to obtain a
path interpolant $\{I_i\}$ with weak interpolants at the beginning
(i.e., $I_1,I_2,\ldots$)  and strong interpolants at the end (i.e.,
$\ldots, I_{n-1}, I_n$). This would increase the likelihood of the sequence
to be inductive and can be achieved by using a family of systems of
different strength.
Another example is software Upgrade Checking, where function summaries 
are computed by interpolation. Different functions in a program could 
require different levels of abstraction by means of interpolation. 
A system that generates stronger interpolants can yield
a tighter abstraction, more closely reflecting the behavior of the 
corresponding function. On the other hand, a system that generates 
weaker interpolants would give an abstraction which is more ``tolerant'' and 
is more likely to remain valid when the function is updated. 

{\em Contribution 2:} This paper systematically studies the collectives 
and the relationships among them; in particular, it shows
that for families of interpolation systems the collectives form a
hierarchy, 
whereas for a single system all but two (i.e.,
path interpolation and simultaneous abstraction) are
equivalent (see \S\ref{sec:interp_family}).

Another issue which this paper deals with is the fact
that there exist different approaches for generating interpolants. 
One is to use specialized algorithms: examples are procedures based on
constraint solving (e.g.,~\cite{RS07}), machine learning
(e.g.,~\cite{SNA12}), and, even, pure verification algorithms like
IC3~\cite{B11} and PDR~\cite{EMB11} that can be viewed as computing a
path interpolation sequence. A second, well-known approach is to extract an
interpolant of $A \land B$ from a resolution proof of unsatisfiability
of $A \land B$.  Examples are the algorithm by Pudl\'ak~\cite{Pud97}
(also independently proposed by Huang~\cite{H95} and by
Kraj\'i\v{c}ek~\cite{Kra97}), the algorithm by McMillan \cite{McM04b},
and the Labeled Interpolation Systems (\emph{LISs}) of D'Silva et
al.~\cite{DKPW10}, the latter being the most general version of this
approach.

The variety of interpolation algorithms makes it difficult 
to reason about their properties in a systematic manner. 
At a low level of representation, the challenge is determined by
the complexity of individual algorithms and by the diversity among them, 
which makes it hard to study them uniformly.
On the other hand, at a high level, where the details are hidden, not 
many interesting results can be obtained.
For this reason, this paper adopts a twofold approach, working both at a high and
at a low level of representation: at the high level, we give a global
view of the entire collection of properties and of their relationships
and hierarchy; at the low level, we obtain additional stronger
results for concrete interpolation systems.
In particular, we first investigate the properties of interpolation
systems treating them as black boxes, and then focus on the propositional LISs.
 In the paper, the results of \S\ref{sec:interp_family} apply to arbitrary
interpolation algorithms, while those of
\S\ref{sec:interp_sing_lab} apply to LISs.

{\em Contribution 3:} For the first time, this paper gives both sufficient and
necessary conditions for a family of LISs and for a single LIS to
enjoy each of the collectives. In particular, we show that in case of
a single system path interpolation is common to all LISs, while
simultaneous abstraction is as strong as all other 
properties. Concrete applications
of our results are also discussed (see \S\ref{sec:interp_sing_lab}).

{\em Contribution 4.} We developed an interpolating prover, PeRIPLO,
implementing the proposed framework as discussed in \S\ref{sec:implementation};
PeRIPLO is currently employed for solving and interpolation by the FunFrog and eVolcheck tools.



\paragraph{\textbf{Related Work}.} 
To our knowledge, despite interpolation being an important component
of verification, no systematic investigation of verification-related
requirements for interpolants has been done prior to this paper. One exception is the work
by the first two authors~\cite{RSS12}, that studies a subset of the
properties in the context of LISs. This paper significantly extends the
results of that work by considering the most common collectives 
used in verification, at the same time addressing a wider class of interpolation systems. Moreover, for
LISs, it provides both the \emph{necessary} and  \emph{sufficient}
conditions for each property.



\vspace*{-2.5mm}
\section{Interpolation Systems}
\label{sec:int_prop}
\vspace*{-2.5mm}

In this section we introduce the basic notions of
interpolation, and then proceed to discuss the collectives.
Among several possible styles of presentation, we chose the one that 
highlights te use of collectives in the context of model checking.
We employ the standard convention of identifying conjunctions of formulae
with sets of formulae and concatenation with conjunction, whenever
convenient. For example, we interchangeably use
$\{\phi_1,\ldots,\phi_n\}$ and $\phi_1 \cdots \phi_n$ for $\phi_1
\wedge \ldots \wedge \phi_n$.

\paragraph{Interpolation System.} 
  An \emph{interpolation system} $Itp_S$ is a function that, given an
  inconsistent ${\Phi=\{\phi_1,\phi_2\}}$, returns a
  \emph{Craig's interpolant}, that is a formula
  $I_{\phi_1,S}= Itp_S(\phi_1  \mid \phi_2)$ such that:
{
\setlength\abovedisplayskip{2mm}
\setlength\belowdisplayskip{2mm}
  \begin{align*}\phi_1 &\implies I_{\phi_1,S} & I_{\phi_1,S} \wedge
    \phi_2 &\implies \bot & \mL_{I_{\phi_1,S}} \subseteq \mL_{\phi_1} \cap
\mL_{\phi_2}
\end{align*}
}
where $\mL_{\phi}$ denotes the atoms of a formula $\phi$.  That is,
$I_{\phi_1,S}$ is implied by $\phi_1$, is inconsistent with $\phi_2$
and is defined over the common language of $\phi_1$ and $\phi_2$.

For $\Phi = \{\phi_1,\ldots,\phi_n\}$, we write $I_{\phi_1\cdots
  \phi_i ,S}$ to denote $Itp_S(\phi_1 \cdots \phi_i \mid \phi_{i+1}
\cdots \phi_n)$.  W.l.o.g., we assume that, for any $Itp_S$ and any formula $\phi$, 
${Itp_S(\top \mid \phi) = \top}$ and $Itp_S(\phi \mid \top) = \bot$,  
where we equate the constant true $\top$ with the empty formula.
We omit $S$ whenever clear from the context.

  An interpolation system $Itp$ is called \emph{symmetric} if for any
  inconsistent \lb $\Phi=\{\phi_1,\phi_2\}$:
$Itp(\phi_1 \mid \phi_2) \iff \n{Itp(\phi_2 \mid \phi_1)}$ 
(we use the notation $\n{\phi}$ for the negation of a formula $\phi$).

  A sequence $\mF = \{Itp_{S_1},\ldots,Itp_{S_n}\}$ of interpolation
  systems is called a \emph{family}.\\ 

\vspace{-3mm}
\paragraph{\textbf{Collectives}.}
In the following, we formulate the properties of interpolation systems that
are required by existing verification algorithms. Furthermore, we
generalize the collectives by presenting them over families of
interpolation systems (i.e., we allow the use different systems to generate different interpolants in a sequence). 
Later, we restrict the properties to the more
traditional setting of the singleton
families. 

\paragraph{\textbf{\emph{$n$-Path Interpolation (PI)}}} 
was first defined in~\cite{JM06}, where it is employed in the refinement
phase of CEGAR-based predicate abstraction. It has also appeared
in~\cite{VG09} under the name \emph{interpolation-sequence}, where it
is used for a specialized interpolation-based hardware verification algorithm. 

Formally, a family of ${n+1}$ interpolation systems $\{Itp_{S_0},\ldots,Itp_{S_n}\}$ has
the \emph{$n$-path} \emph{interpolation} property ($n$-PI) iff for any
inconsistent $\Phi = \{\phi_1,\ldots,\phi_n\}$ and for $0\leq
i \leq n-1$ (recall that $I_{\top}= \top$ and $I_{\Phi} =\bot$):
\[(I_{\phi_1\ldots \phi_i,S_i} \wedge \phi_{i+1}) \implies
I_{\phi_1\ldots \phi_{i+1},S_{i+1}}\]

\paragraph{\textbf{\emph{$n$-Generalized Simultaneous Abstraction (GSA)}}} is the generalization 
of \emph{simultaneous abstraction}, a property that first appeared, under the name \emph{symmetric interpolation}, in~\cite{JM05},
where it is used for approximation of
a transition relation for predicate abstraction. We changed the name
to avoid confusion with the notion of \emph{symmetric interpolation
  system} (see above). The reason for generalizing
the property will be apparent later.

Formally, a family of ${n+1}$ interpolation systems $\{Itp_{S_1},\ldots,Itp_{S_{n+1}}\}$ has the \emph{$n$-generalized simultaneous abstraction}  
property (\mbox{$n$-GSA}) iff for 
any inconsistent $\Phi = \{\phi_1,\ldots,\phi_{n+1}\}$:
\[\bigwedge_{i=1}^{n} I_{\phi_i,S_i} \implies I_{\phi_1\ldots \phi_n,S_{n+1}}\]
The case $n=2$ is called \emph{Binary GSA (BGSA)}:
$I_{\phi_1,S_1} \wedge I_{\phi_2,S_2} \implies I_{\phi_1 \phi_2,S_3}$.\\
If $\phi_{n+1} = \top$, the property is called \emph{$n$-simultaneous
  abstraction} (\mbox{$n$-SA}):\lb
$\bigwedge_{i=1}^{n} I_{\phi_i,S_i} \implies \bot (= I_{\phi_1\ldots \phi_n,S_{n+1}})$ 
and, if $n=2$, \emph{binary SA (BSA)}. In $n$-SA $Itp_{S_{n+1}}$ is
irrelevant and is often omitted.

\paragraph{\textbf{\emph{$n$-State-Transition Interpolation (STI)}}}  
is defined as a combination of PI and SA in a single
family of systems. It was introduced in ~\cite{AGC12a} as part of
the inter-procedural verification algorithm
\textsc{Whale}. Intuitively, the ``state'' interpolants
over-approximate the set of reachable states, and the ``transition''
interpolants summarize the transition relations (or function
bodies). The STI requirement ensures that state over-approximation is
``compatible'' with the summarization. That is, $\{I_{\phi_1 \cdots
  \phi_i,S_i} \} I_{\phi_{i+1},T_{i+1}} \{I_{\phi_1 \cdots
  \phi_{i+1},S_{i+1}}\}$ is a valid Hoare triple for each $i$.
  
  Formally, a family of interpolation systems
  $\{Itp_{S_0},\ldots,Itp_{S_n},Itp_{T_1},\ldots,Itp_{T_n}\}$ has the
  \emph{$n$-state-transition interpolation} property ($n$-STI) iff for
  any inconsistent $\Phi = \{\phi_1,\ldots,\phi_n\}$ and for $0\leq i \leq n-1$:
  \[(I_{\phi_1\ldots \phi_i,S_i} \wedge I_{\phi_{i+1},T_{i+1}})
  \implies I_{\phi_1\ldots \phi_{i+1},S_{i+1}}\]

\paragraph{\textbf{\emph{$T$-Tree Interpolation (TI)}}} is a generalization of classical 
interpolation used in model checking applications, in which
partitions of an unsatisfiable formula naturally correspond to a tree structure such as call tree or program unwinding.
The collective was first introduced by McMillan and
Rybalchenko for computing post-fixpoints of a system of Horn clauses
(e.g., used in analysis of recursive programs)~\cite{duality}, and is 
 equivalent to the nested-interpolants of~\cite{HHP10}.

Formally, let $T = (V, E)$ be a tree with $n$ nodes $V = [1,\ldots,n]$. A family of
  $n$ interpolation systems $\{Itp_{S_1}, \ldots, Itp_{S_n}\}$ has the
  \emph{$T$-tree interpolation property} ($T$-TI) iff for any inconsistent $\Phi =
  \{\phi_1, \ldots, \phi_n\}$:
  \[
  \bigwedge_{(i,j) \in E} I_{F_j,S_j}  \wedge \phi_i \implies I_{F_i,S_i} 
  \]
  where $F_i = \{ \phi_j \mid i \sqsubseteq j \}$,
  and $i \sqsubseteq j$ iff node $j$ is a descendant of node $i$ in
  $T$. Notice that for the root $i$ of $T$, $F_i = \Phi$ and
  $I_{F_i,S_i} = \bot$.

An interpolation system $Itp_S$ is said to \emph{have a property} $P$ (or,
simply, to have $P$), where $P$ is one of the properties defined
above, if every family induced by $Itp_S$ has $P$. For example,
$Itp_S$ has GSA iff for every $k$ the family $\{Itp_{S_1}, \ldots,
Itp_{S_k}\}$, where $Itp_{S_i} = Itp_S$ for all $i$, has $k$-GSA.

%



\vspace{-1mm}
\section{Collectives of Interpolation Systems}
\label{sec:interp_sing}
\vspace{-1mm}

In this section, we study collectives of general interpolation
systems, that is, we treat interpolation systems as black-boxes.
In section \S\ref{sec:lis}  we will extend the study to the implementation-level 
details of the LISs. 

\paragraph{Collectives of Single Systems.}
We begin by studying the relationships among the various collectives
of single interpolation systems.

\begin{theorem}
  \label{thm:single-prop-collapse}
  Let $Itp_S$ be an interpolation system. The following are
  equivalent: $Itp_S$ has BGSA (1), $Itp_S$ has GSA (2), $Itp_S$ has
  TI (3), $Itp_S$ has STI (4).
\end{theorem}
\begin{proof}
  We show that $1 \to 2$, $2 \to 3$, $3 \to
  4$, $4 \to 1$.

$(1 \to 2)$ Assume $Itp_S$ has BGSA. Take any
  inconsistent $\Phi = \{\phi_1,\ldots,\phi_{n+1}\}$.  Then,
  for $2 \leq i \leq n$:
$
(I_{\phi_1 \cdots \phi_{i-1}} \land I_{\phi_i}) \Rightarrow I_{\phi_1 \cdots \phi_{i}} 
$, 
which together yield
$(\bigwedge_{i=1}^{n} I_{\phi_i}) \Rightarrow I_{\phi_1\ldots \phi_n}$.
Hence, $Itp_S$ has GSA.

$(2\to 3)$ Let $T= ([1,\ldots,n],E)$, take any
inconsistent $\Phi = \{\phi_1,\ldots,\phi_{n}\}$.  Since
$Itp_S$ has GSA:
$
(\bigwedge_{(i,j) \in E} I_{F_j} \wedge I_{\phi_i})
\Rightarrow I_{F_i}
$, 
and, from the definition of Craig interpolation, $\phi_i \Rightarrow
I_{\phi_i}$. Hence, $Itp_S$ has $T$-TI.

$(3 \to 4)$ Take any inconsistent $\Phi =
\{\phi_1,\ldots,\phi_{n}\}$ and extend it to a $\Phi'$ by adding
$n$ copies of $\top$ at the end. Define a tree $T_{STI} =
([1,\ldots,2n],E)$ s.t.: 
$
  E = \{ (n+i, i) \mid 1\leq i \leq n \} \cup
  \{(n+i,n+i-1) \mid 1\leq i \leq n\}
$. 
Then, for $1 \leq i \leq n$, $F_i = \{ \phi_i \}$ and  $F_{n+i} = \{\phi_1, \ldots, \phi_i\}$, where $F_i$ is as in the definition of
$T$-TI. By the $T$-TI property: 
$ 
(I_{F_{n + i}} \land I_{F_{i+1}} \land \top) \Rightarrow I_{F_{n+i+1}}
$, 
which is equivalent to STI.

$(4 \to 1)$ Follows from STI being syntactically
equivalent to BGSA for $i=1$.  
\end{proof}


Theorem~\ref{thm:single-prop-collapse} has a few simple
extensions. First, $GSA$ implies $SA$ directly from the
definitions. Similarly, since 
$\phi \Rightarrow I_\phi$, STI implies PI.  Finally, we conjecture
that both SA and PI are strictly weaker than the rest.  In
\S\ref{sec:interp_sing_lab} (Theorem~\ref{theo:pi_not_bgsa}), we show
that for LISs, PI is strictly weaker than SA. As for SA, we show that
it is equivalent to BGSA in symmetric interpolation systems
(Proposition~\ref{lem:single_sa_to_bgsa} in the appendix). But, in the
general case, the conjecture remains open.


These results define a hierarchy of collectives which is summarized in
Fig.~\ref{fig:prop-single}, where the edges indicate implications
among the collectives. Note that $SA \to GSA$ holds only for symmetric
systems.

In summary, the main contribution in the setting of a single system is the proof that  almost all collectives are
equivalent and the hierarchy of the collectives collapses. From a
practical perspective, this means that McMillan's interpolation system
(implemented by most interpolating SMT-solvers) has all of the
collective properties, including the recently introduced TI.
\\


\paragraph{Collectives of Families of Systems.}
\label{sec:interp_family}
Here, we study collectives of families of interpolation systems. We
first show that the collectives introduced in \S\ref{sec:int_prop}
directly extend from families to sub-families. Second, we examine the
hierarchy of the relationships among the properties. Finally, we
conclude by discussing the practical implications of these results.

\paragraph{Collectives of Sub-families.}
If a family of interpolation systems $\mathcal{F}$ has a property $P$,
then sub-families of $\mathcal{F}$ have $P$ as well. We state
this formally for $k$-STI (since we use it in the proof of
Theorem~\ref{thm:sa_pi_not_sti}); similar statements for the other collectives are
discussed in the appendix\footnote{All proofs can be found in the appendix.}.
\begin{theorem}
\label{lem:sti_sub}
A family $\{Itp_{S_0},\ldots,Itp_{S_n},Itp_{T_1},\ldots,Itp_{T_n}\}$
has $n$-STI iff for all $k\leq n$ the sub-family
$\{Itp_{S_{0}},\ldots,Itp_{S_{k}}\}$ $\cup$
$\{Itp_{T_{1}},\ldots,Itp_{T_{k}}\}$ has $k$-STI.
\end{theorem}

\paragraph{Relationships Among Collectives.}
We now show the relationships among collectives.  First, we note
that $n$-SA and BGSA are equivalent for symmetric interpolation
systems. Whenever a family $\mF = \{Itp_{S_1},\ldots,Itp_{S_{n+1}}\}$
has \mbox{$(n+1)$-SA} and $Itp_{S_{n+1}}$ is symmetric, then $\mF$ has
$n$-GSA (Proposition~\ref{prop:sa_to_bgsa} in the appendix, which is
the analogue of Proposition~\ref{lem:single_sa_to_bgsa} for single
systems).

In the rest of the section, we delineate the hierarchy of
collectives. In particular, we show that $T$-TI is the most general
collective, immediately followed by $n$-GSA, which is followed by
$BGSA$ and $n$-STI, which are equivalent, and at last by $n$-SA and
$n$-PI.
%
The first result is that the $n$-STI property implies both the $n$-PI
and $n$-SA properties separately:
\begin{theorem}
\label{lem:st_to_pi_sa}
If a family $\mF = 
\{Itp_{S_0},\ldots,Itp_{S_n},Itp_{T_1},\ldots,Itp_{T_n}\}$ has
$n$-STI then (1) $\{Itp_{S_0},\ldots,Itp_{S_n}\}$ has
$n$-PI and (2) $\{Itp_{T_1},\ldots,Itp_{T_n}\}$ has 
$n$-SA.
\end{theorem}
A natural question to ask is whether the converse of
Theorem~\ref{lem:st_to_pi_sa} is true. That is, whether the family
$\mF_1 \cup \mF_2$ that combines two arbitrary families $\mF_1$ and
$\mF_2$ that independently enjoy $n$-PI and $n$-SA, respectively, has
$n$-STI.  We show in \S\ref{sec:lis}, Theorem~\ref{thm:sa_pi_not_sti},
that this is not the case.

As for BGSA, the $n$-STI property is closely related to it: deciding
whether a family $\mF$ has $n$-STI is in fact reducible to deciding
whether a collection of sub-families of $\mF$ has BGSA.
\begin{theorem}
\label{lem:st-and-gsa}
A family $\mF = \{Itp_{S_0},\ldots,Itp_{S_n},Itp_{T_1},\ldots,Itp_{T_n}\}$
has $n$-STI iff
$\{Itp_{S_i},Itp_{T_{i+1}},Itp_{S_{i+1}}\}$ has BGSA for
all  $0\leq i \leq n-1$.
\end{theorem}
From Theorem~\ref{lem:st-and-gsa} and Theorem~\ref{lem:st_to_pi_sa} we derive:
\begin{corollary}
\label{lem:bgsa_to_sa}
If there exists a family $\{Itp_{S_0},\ldots,Itp_{S_n}\}$ $\cup$ $\{Itp_{T_1},\ldots,Itp_{T_n}\}$ s.t. 
 $\{Itp_{S_i},Itp_{T_{i+1}},Itp_{S_{i+1}}\}$ has BGSA for all ${0\leq i \leq n-1}$, 
then  $\{Itp_{T_1},\ldots,Itp_{T_n}\}$  has $n$-SA.
\end{corollary}
We now relate $T$-TI and $n$-GSA.
Note that the need for two theorems with different statements
arises from the asymmetry between the two properties: all  $\phi_i$ 
are abstracted by interpolation in $n$-GSA, whereas in $T$-TI
a formula is not abstracted, when considering the correspondent parent
together with its children.

\begin{theorem}
\label{lem:ti_to_gsa}
  Given a tree $T = (V,E)$ if a family ${\mF = \{Itp_{S_i}\}_{i \in V}}$ has $T$-TI, then, for every
  parent $i_{k+1}$ and its children $i_1,\ldots,i_{k}$:
  \begin{enumerate}
  \item If $i_{k+1}$ is the root,
    $\{Itp_{S_{i_1}},\ldots,Itp_{S_{i_{k}}}\}$ has $k$-SA.
  \item Otherwise,
    $\{Itp_{S_{i_1}},\ldots,Itp_{S_{i_{k}}},Itp_{S_{i_{k+1}}}\}$ has
    $k$-GSA.
  \end{enumerate}
\end{theorem}

\begin{theorem}
\label{lem:gsa_to_ti}
 Given a  tree $T = (V,E)$, a family $\mF = \{Itp_{S_i}\}_{i \in V}$ has $T$-TI
  if, for every node $i_{k+1}$ and its children $i_1,\ldots,i_{k}$, there exists $T_{i_{k+1}}$ such that:
  \begin{enumerate}
  \item If $i_{k+1}$ is the root,
    $\{Itp_{S_{i_1}},\ldots,Itp_{S_{i_{k}}},Itp_{T_{i_{k+1}}}\}$ has ${(k+1)}$-SA.
  \item Otherwise, 
   $\{Itp_{S_{i_1}},\ldots,Itp_{T_{i_{k+1}}},Itp_{S_{i_{k+1}}}\}$ has \mbox{${(k+1)}$-GSA}.
  \end{enumerate}
\end{theorem}
An important observation is that the $T$-TI property is the most
general, in the sense that it realizes any of the other properties,
given an appropriate choice of the tree $T$.  We state here (and prove in the appendix) that $n$-GSA and
$n$-STI can be implemented by $T$-TI for some $T_{GSA}^n$ and $T^n_{STI}$; the remaining cases can be
derived in a similar manner. Note that the converse implications are not necessarily true in
general, since the tree interpolation requirement is stronger.
\begin{theorem}
\label{theo:tgsa}
  If a family $\mF = \{Itp_{S_{n+1}}, Itp_{S_1}, \ldots,
  Itp_{S_{n+1}}\}$ has $T_{GSA}^n$-TI, then $\{Itp_{S_1},
  \ldots, Itp_{S_{n+1}}\}$ has $n$-GSA.
\end{theorem}
  \begin{theorem}
\label{theo:tsti}
  If a family $\mF =
  \{Itp_{S_0},\ldots,Itp_{S_n}\}\cup{}$ $\{Itp_{T_1},\ldots,Itp_{T_n}\}$ has 
  $T^n_{STI}$-TI, then it has $n$-STI.
\end{theorem}
%
%
%
%
%
The results of so far (including
Theorem~\ref{thm:sa_pi_not_sti} of \S\ref{sec:lis}) define a hierarchy of collectives which is summarized in
Fig.~\ref{fig:family-sum}. The solid edges indicate direct implication between
properties; $SA \to GSA$ requires symmetry, while $GSA \to
TI$ requires the existence of an additional set of interpolation systems.
The dashed edges represent the ability of $TI$ to realize all the
other properties for an appropriate tree; only the edges to $STI$ and $GSA$ are shown,
the other ones are implicit. The dash-dotted edges represent the sub-family properties.

An immediate application of our
results is that they show how to overcome limitations of existing
implementations. For example, they enable the trivial
construction of tree interpolants in MathSat\footnote{http://mathsat.fbk.eu/} (currently only available
in iZ3) -- thus enabling its usability for Upgrade Checking~\cite{SFS12} -- by reusing existing
BGSA-interpolation implementation of MathSat. Similarly, our results
enable construction of BGSA and GSA interpolants in iZ3 (currently
only available in MathSat) -- thus enabling the use of iZ3 in
Whale.
\begin{figure}[ht]
  \vspace{-2cm}
\begin{minipage}{0.45\textwidth}
  \centering
    \vspace{2cm}
\begin{tikzpicture}[node distance=1cm,style={font=\sffamily\small}]

  \node (1) {BGSA};
  \node (2) [below left of=1] {GSA};
  \node (3) [below right of=2] {TI};
  \node (4) [below right of=1] {STI};
  \node (5) [below right of=4] {PI};
  \node (6) [below left of=2] {SA};

  \path[->,every node/.style={font=\sffamily\small}]
    (1) edge [bend right] node {} (2)
    (2) edge [bend right] node {} (3)
	edge [bend left] node  {} (6)	
    (3) edge [bend right] node {} (4)
    (4) edge [bend right] node {} (1)
        edge node   {} (5)
    (6) edge [bend left] node[left] {\scriptsize{symm}} (2);    
\end{tikzpicture}  
  \caption{Single systems collectives.}
  \label{fig:prop-single}
\end{minipage}
\hspace{1mm}
\begin{minipage}{0.51\textwidth}
  \centering
\begin{tikzpicture}[node distance=1cm,style={font=\sffamily\small}]
  \node (1) {};
  \node (2) [right of=1] {};
  \node (3) [right of=2] {TI};
  \node (4) [right of=3] {};
  \node (5) [right of=4] {};
  \node (6) [below of=1] {};
  \node (7) [right of=6] {GSA};
  \node (8) [right of=7] {};
  \node (9) [right of=8] {STI};
  \node (10) [right of=9] {};
  \node (11) [below of=6] {BGSA};
  \node (12) [right of=11] {};
  \node (13) [right of=12] {SA};
  \node (14) [right of=13] {};
  \node (15) [right of=14] {PI};

  \path[->,every node/.style={font=\sffamily\small}]
    (3) edge [bend left,looseness=1] node {} (7)
    (7) edge [bend left,looseness=1] node[above left=-4.5mm] {\Large{*}} (3)
	edge [bend left] node {} (13)
	edge node {} (11)
    (9) edge node {} (13)
        edge node {} (15)
    (13) edge [bend left] node[below left=-2mm] {\scriptsize{symm}} (7);
    
    \path[->,densely dashdotted,every node/.style={font=\sffamily\small}]
    (3)	edge [loop above] node {} (3)
    (7)	edge [in=170,out=140,loop] node {} (7)
    (13) edge [loop right] node {} (13)
    (9)	 edge [in=10,out=-20,loop] node {} (9)
    (15) edge [loop right] node {} (15);
    
    \path[->,loosely dashed,every node/.style={font=\sffamily\small}]
        (3) edge [] node {} (7)
	edge node {} (9);
    
    \path[<->] (11) edge  [in=60,out=120,looseness=2.1] node {} (9);    
\end{tikzpicture}  
  \caption{Families of systems collectives.}
  \label{fig:family-sum}
 \end{minipage}
   \vspace{-1cm}
\end{figure}



\section{Collectives of Labeled Interpolation Systems}
\label{sec:lis}
\vspace{-2mm}

In this section, we move from the abstract level of general interpolation systems to the 
implementation level of the propositional Labeled Interpolation Systems. After introducing and defining LISs, we 
study collectives of families, then summarize the results for single LISs, also
answering the questions left open in \S\ref{sec:interp_sing}.
The key results are in Lemmas $1-4$. Unfortunately, the proofs are quite
technical. For readability, we focus on the main results and their
significance and refer the reader to the appendix for full details.

\noindent{}There are several state-of-the art approaches for automatically
computing interpolants. The most successful techniques derive an
interpolant for $A \wedge B$ from a resolution proof of the
unsatisfiability of the conjunction.  Noteworthy examples are the
algorithm independently developed by Pudl\'ak \cite{Pud97}, Huang
\cite{H95} and Kraj\'i\v{c}ek \cite{Kra97}, and the one by McMillan~\cite{McM04b}.  These algorithms are implemented recursively by
initially computing \emph{partial interpolants} for the axioms (leaves
of the proof), and, then, following the proof structure,
by computing a partial interpolant for each conclusion from those of the
premises.  The partial interpolant of the root of the proof is the
interpolant for the formula. In this section, we review these
algorithms following the framework of D'Silva et al.~\cite{DKPW10}.\\

\vspace{-3mm}
\paragraph{\textbf{Resolution Proofs}.}
We assume a \emph{countable set} of propositional variables.  A
\emph{literal} is a variable, either with positive ($p$) or negative
($\n{p}$) polarity.  A \emph{clause} $C$ is a finite disjunction of
literals; a formula $\Phi$ in conjunctive normal form (CNF) is a
finite conjunction of clauses.  A \emph{resolution proof of
  unsatisfiability} (or \emph{refutation}) of a formula $\Phi$ in CNF
is a tree such that the leaves are the clauses of $\Phi$, the root is
the empty clause $\bot$ and the inner nodes are clauses generated via
the \emph{resolution rule} (where $C^+ \vee p$ and $C^- \vee \n{p}$
are the \emph{antecedents}, $C^+ \vee C^-$ the \emph{resolvent}, and
$p$ is the \emph{pivot}): \vspace{-3mm} {\small
\[ \frac{C^+ \vee p \qquad  C^- \vee \n{p}}{C^+ \vee C^-} \] 
}

\vspace{-3mm}
\paragraph{\textbf{Labelings and Interpolant Strength}.}
D'Silva et al.~\cite{DKPW10} generalize the algorithms by Pudl\'ak~\cite{Pud97} and McMillan~\cite{McM04b}
for propositional resolution systems by introducing the notion of \emph{Labeled Interpolation System} (LIS), 
focusing on the concept of \emph{interpolant strength} (a formula $\phi$ is stronger than $\psi$ when $\phi \! \! \! \implies \! \! \! \psi$). 

Given a refutation of a formula $A \wedge B$, a variable $p$ can appear as a literal only in $A$, only in $B$ or in both; $p$
is respectively said to have \emph{class} $A$, $B$ or $AB$.
A \emph{labeling} $L$ is a mapping that assigns a \emph{label} among $\{a,b,ab\}$ independently to each variable
in each clause; we assume that no clause has both a literal and its negation, so assigning a label to variables or literals is equivalent. 
The set of possible labelings is restricted by ensuring that class $A$ variables have label $a$ and 
class $B$ variables label $b$; $AB$ variables can  be labeled either $a$, $b$ or $ab$. 

In \cite{DKPW10}, a \emph{Labeled Interpolation System} (LIS) is defined as a procedure $Itp_L$ (shown in Fig.~\ref{tab:gen}) that, given $A$, $B$, a 
refutation $R$ of $A \wedge B$ and a labeling $L$, outputs 
a partial interpolant $I_{A,L}(C) = Itp_{L}(A \mid B)(C)$ for any clause $C$ in $R$; 
this depends on the clause being in $A$ or $B$ (if leaf) and on the label of the pivot
associated with the resolution step (if inner node). $I_{A,L}=  Itp_{L}(A \mid B)$ represents the interpolant for $A\wedge B$, 
that is $Itp_{L}(A \mid B)(\bot)$.
We omit the parameters whenever clear from the context. 

\begin{figure}[t]
\centering
\vspace{-1.3cm}
\begin{tabular}{|l c|c|c|c || l c|c|c|c|}%
\hline
Leaf: & \multicolumn{4}{c|}{$C \, [I]$} & Inner node: & \multicolumn{4}{c|}{$\quad \dfrac{C^+ \vee p:\alpha \, [I^+] \qquad  C^- \vee \n{p}:\beta \, [I^-]}{C^+ \vee C^- \, [I]}$} \\
\hline 
\multicolumn{5}{|c|}{$I =
\left\{
	\begin{array}{ll}
		C \!\!\downharpoonright b  & \quad \mbox{if } C \in A\\
		\neg (C \!\!\downharpoonright a) & \quad \mbox{if } C \in B\\
	\end{array}
\right. $
}
&
\multicolumn{5}{|c|}{$I =
\left\{
	\begin{array}{ll}
		I^+ \vee I^-  & \quad \mbox{if } \alpha \sqcup \beta = a\\
		I^+ \wedge I^- & \quad \mbox{if } \alpha \sqcup \beta = b\\
		(I^+ \vee p) \wedge (I^- \vee \n{p}) & \quad \mbox{if } \alpha \sqcup \beta = ab 
	\end{array}
\right. $
}\\
\hline 
\end{tabular}
\vspace{-0.14in}
\caption{Labeled Interpolation System $Itp_L$.}
\label{tab:gen}
\vspace{-0.32in}
\end{figure}

In Fig.~\ref{tab:gen}, $C \!\!\downharpoonright \alpha$ denotes the restriction of a clause 
$C$ to the variables with label $\alpha$. $p:\alpha$ indicates that variable $p$ has label $\alpha \in \{a,b,ab\}.$
By $C[I]$ we represent that clause $C$ has a partial interpolant $I$.
$I^+$, $I^-$ and $I$ are the partial interpolants respectively associated with
the two antecedents and the resolvent of a resolution step: $I^+ =
Itp_L(C^+ \vee p)$, $I^- = Itp_L(C^- \vee \n{p})$, $I = Itp_L(C^+ \vee C^-)$. 

A join operator $\sqcup$ allows to determine the label of a pivot $p$, taking into account
that $p$ might have different labels $\alpha$ and $\beta$ in the two antecedents: 
$\sqcup$ is defined by $a \sqcup b = ab$, $a \sqcup ab = ab$,
$b \sqcup ab = ab$.

The systems corresponding to McMillan and Pudl\'ak's interpolation algorithms are referred to as $Itp_{M}$ and $Itp_{P}$; the system dual
to McMillan's is $Itp_{M'}$.  $Itp_{M}$, $Itp_{P}$ and $Itp_{M'}$ are
obtained as special cases of $Itp_L$ by labeling all the occurrences of $AB$ variables with $b$, $ab$ and $a$, respectively 
(see \cite{DKPW10} and \cite{RSS12}).

A total order $\preceq$ is defined over labels as ${b\preceq ab \preceq a}$, and pointwise extended to a partial order over labelings: $L \preceq L'$ 
if, for every clause $C$
and variable $p$ in $C$, $L(p,C) \preceq L'(p,C)$. This allows to directly compare the logical strength of the interpolants produced by two systems. 
In fact, for any refutation $R$ of a formula $A\wedge B$ and labelings $L,L'$ such that $L \preceq L'$, we have $Itp_L(A,B,R) \implies Itp_{L'}(A,B,R)$
 and we say that $Itp_L$ is \emph{stronger} than $Itp_{L'}$~\cite{DKPW10}.  


Since a labeled system $Itp_L$ is uniquely determined by the labeling $L$, when discussing a family of LISs $\{Itp_{L_1},\ldots,Itp_{L_n}\}$
we will refer to the correspondent \emph{family of labelings} as  $\{L_1,\ldots,L_n\}$.\\

\vspace*{-0.11in}
\paragraph{\textbf{Labeling Notation}.}
In the previous sections, we saw how the various collectives involve the generation of multiple interpolants from a single inconsistent formula
$\Phi=\{\phi_1,\ldots,\phi_n\}$ for different subdivisions of $\Phi$ into an $A$ and a $B$ parts; we  refer to these
ways of splitting $\Phi$ as \emph{configurations}. Remember that a labeling $L$ has freedom in assigning labels only to occurrences
of variables of class $AB$; each configuration identifies these variables. 

Since we deal with several configurations at a time, it is useful to
separate the variables into \emph{partitions} of $\Phi$ depending on
whether the variables are local to a $\phi_i$ or shared, taking into
account all possible combinations.  For example, Table~\ref{tab:3sa}
is the \emph{labeling table} that characterizes $3$-SA. Recall that in
3-SA we are given an inconsistent $\Phi =\{\phi_1,\phi_2,\phi_3\}$ and
a family of labelings $\{L_1,L_2,L_3\}$ and generate three
interpolants $I_{\phi_1,L_1}$, $I_{\phi_2,L_2}$, $I_{\phi_3,L_3}$. The
labeling $L_i$ is associated with the $i$th configuration. For
example, the table shows that $L_1$ can independently assign a label
from $\{a,b,ab\}$ to each occurrence of each variable shared between
$\phi_1$ and $\phi_2$, $\phi_1$ and $\phi_3$ or $\phi_1,\phi_2$ and
$\phi_3$ (as indicated by the presence of $\alpha_1,\gamma_1,\delta_1$).

When talking about an occurrence of a variable $p$ in a certain
partition $\phi_{i_1}\cdots\phi_{i_k}$, it is convenient to associate
to $p$ and the partition a \emph{labeling vector}
$(\eta_{i_1},\ldots,\eta_{i_k})$, representing the labels assigned to
$p$ by $L_{i_1},\ldots,L_{i_k}$ in configuration $i_1,\ldots,i_k$ (all
other labels are fixed). Strength of labeling vectors is compared
pointwise, extending the linear order $b\preceq ab\preceq a$ as
described earlier.
\begin{table}[t]
\vspace{-0.5cm}
\begin{minipage}{60mm}%
\centering%
\begin{tabular}{|l|c|c|c|c|}%
\hline 
\multicolumn{1}{|c|}{\multirow{2}{*}{$p$ in ?}} & \multicolumn{3}{c|}{Variable $class$, $label$}\\
\cline{2-4}
&$\phi_1 \mid \phi_2 \phi_3$  & $\phi_2 \mid \phi_1 \phi_3 $ & $\phi_3 \mid \phi_1 \phi_2$\\
\hline
$\phi_1$ & $A,a$ & $B,b$ & $B,b$\\
$\phi_2$ & $B,b$ & $A,a$ & $B,b$\\
$\phi_3$ & $B,b$ & $B,b$ & $A,a$\\
$\phi_1 \phi_2$ & $AB,\alpha_1$ & $AB,\alpha_2$ & $B,b$\\
$\phi_2 \phi_3$ & $B,b$ & $AB,\beta_2$ & $AB,\beta_3$\\
$\phi_1 \phi_3$ & $AB,\gamma_1$ & $B,b$ & $AB,\gamma_3$\\
$\phi_1 \phi_2 \phi_3$ & $AB,\delta_1$ & $AB,\delta_2$ & $AB,\delta_3$\\
\hline%
\end{tabular}
\caption{$3$-SA.}%
\label{tab:3sa}%
\end{minipage}
\begin{minipage}{60mm}%
\centering%
\begin{tabular}{|l|c|c|c|}%
\hline 
\multicolumn{1}{|c|}{\multirow{2}{*}{$p$ in ?}} & \multicolumn{3}{c|}{Variable $class$, $label$}\\
\cline{2-4}
&$\phi_1 \mid \phi_2 \phi_3$ & $\phi_2 \mid \phi_1 \phi_3$ & $\phi_1 \phi_2 | \phi_3$ \\
\hline
$\phi_1$ & $A,a$ & $B,b$ & $A,a$ \\
$\phi_2$ & $B,b$ & $A,a$ & $A,a$ \\
$\phi_3$ & $B,b$ & $B,b$ & $B,b$\\
$\phi_1 \phi_2$ & $AB,\alpha_1$ & $AB,\alpha_2$ & $A,a$ \\
$\phi_2 \phi_3$ & $B,b$ & $AB,\beta_2$ & $AB,\beta_3$ \\
$\phi_1 \phi_3$ & $AB,\gamma_1$ & $B,b$ & $AB,\gamma_3$ \\
$\phi_1 \phi_2 \phi_3$ & $AB,\delta_1$ & $AB,\delta_2$ & $AB,\delta_3$ \\
\hline%
\end{tabular}%
\caption{BGSA.}%
\label{tab:bgsa}%
\end{minipage}%
\vspace{-0.4in}
\end{table}
%


%
%
We reduce the problem of deciding whether a family $\mF=\{Itp_{L_1},\ldots,Itp_{L_n}\}$ has an
interpolation property $P$ to showing that all labeling vectors of 
$\{L_1,\ldots,L_n\}$ satisfy a certain set of \emph{labeling constraints}.
For simplicity of presentation, in the rest of the paper we assume
that all occurrences of a variable are labeled uniformly. The
extension to differently labeled occurrences is straightforward.\\

\vspace*{-2mm}
\paragraph{\textbf{Collectives of LISs Families}.}
\label{sec:interp_family_lab}
We derive in the following both \emph{necessary} and \emph{sufficient} conditions for the 
collectives to hold in the context of LISs families. 
The practical significance of our results is to identify which LISs satisfy which collectives. 
In particular, for the first time, we show that not all LISs identified by D'Silva et al. satisfy all collectives. 
This work provides an essential guide for using interpolant strength results when collectives are required (such as in Upgrade Checking). 

We proceed as follows. First, we identify necessary and sufficient
labeling constraints to characterize BGSA. Second, we extend them to
$n$-GSA and to $n$-SA. Third, we exploit the connections between BGSA and
$n$-GSA on one side, and $n$-STI and $T$-TI on the other
(Theorem~\ref{lem:st-and-gsa}, Lemma~\ref{lem:ti_to_gsa},
Lemma~\ref{lem:gsa_to_ti}) to derive the labeling
constraints both for $n$-STI and $T$-TI, thus completing the picture.\\
\paragraph{\textbf{BGSA}.}
Let $\Phi = \{\phi_1, \phi_2, \phi_3\}$ be an unsatisfiable
formula in CNF, and $\mF = \{Itp_{L_1},Itp_{L_2},Itp_{L_3}\}$ a
family of LISs.  We want to identify the
restrictions on the labeling vectors of $\{L_1,L_2,L_3\}$
for which $\mF$ has BGSA, i.e., $I_{\phi_1,L_1} \wedge
I_{\phi_2,L_2}\implies I_{\phi_1 \phi_2, L_3}$.
We define a set of \emph{BGSA constraints} $CC_{BGSA}$ on labelings as follows.
A family of labelings
  $\{L_1,L_2,L_3\}$ satisfies  $CC_{BGSA}$ iff:
{
\setlength{\abovedisplayskip}{5pt}
\setlength{\belowdisplayskip}{5pt}
\[(\alpha_1,\alpha_2), (\delta_1,\delta_2) \preceq \{(ab,ab),(b,a),(a,b)\}, 
 \beta_2 \preceq \beta_3, \gamma_1 \preceq \gamma_3, \delta_1 \preceq \delta_3 ,\delta_2 \preceq \delta_3\]
 }
hold for all variables,
where $\alpha_i$, $\beta_i$, $\gamma_i$ and $\delta_i$ are as shown in
Table~\ref{tab:bgsa}, the labeling table for $BGSA$.  
$*\preceq \{*_1, *_2\}$ denotes that $*\preceq *_1$ or $*\preceq *_2$ (both can be true).

We aim to prove that $CC_{BGSA}$ is  necessary and sufficient for a family of
LISs to have BGSA. On one hand, we claim that, if $\{L_1,L_2,L_3\}$ satisfies $CC_{BGSA}$, then $\{Itp_{L_1},Itp_{L_2},Itp_{L_3}\}$ has BGSA.
It is sufficient to prove the thesis for a set of \emph{restricted BGSA constraints} $CC_{BGSA}^*$, defined as follows:
{
\setlength{\abovedisplayskip}{5pt}
\setlength{\belowdisplayskip}{5pt}
\[
(\alpha_1,\alpha_2), (\delta_1,\delta_2) \in \{(ab,ab),(b,a),(a,b)\}, 
\beta_2 = \beta_3, \gamma_1 = \gamma_3, \delta_3 = \max\{\delta_1,\delta_2\}
 \]
}
\begin{lemma}
\label{lem:bgsa_suff}
If $\{L_1,L_2,L_3\}$ satisfies $CC_{BGSA}^*$, then $\{Itp_{L_1},Itp_{L_2},Itp_{L_3}\}$ has BGSA. 
\end{lemma}
\vspace{-1mm}The $CC_{BGSA}^*$ constraints can be relaxed to $CC_{BGSA}$ as shown in \cite{RSS12} (Theorem~2, Lemma~3), 
due to the connection between partial order on labelings and LISs
and strength of the generated interpolants. 
For example, the constraint $\delta_3 = \max (\delta_1, \delta_2)$ can be relaxed
to $\delta_3 \succeq \delta_1$, $\delta_3 \succeq \delta_2$.
This leads to:
\begin{corollary}
\mbox{If $\{L_1,L_2,L_3\}$ satisfies $CC_{BGSA}$, then  $\{Itp_{L_1},Itp_{L_2},Itp_{L_3}\}$ has BGSA. }
\end{corollary}
\vspace{-1mm}On the other hand, it holds that the satisfaction of the $CC_{BGSA}$ constraints is necessary for BGSA: 
\begin{lemma}
\label{lem:bgsa_nec}
 \mbox{If $\{Itp_{L_1},Itp_{L_2},Itp_{L_3}\}$ has BGSA, then $\{L_1,L_2,L_3\}$ satisfies $CC_{BGSA}$.}
\end{lemma}
\vspace{-1mm}Having proved that $CC_{BGSA}$ is both sufficient and necessary, we conclude:
\begin{theorem}
\label{theo:bgsa_nec_suff}
A family $\{Itp_{L_1},Itp_{L_2},Itp_{L_3}\}$ has  BGSA  if and only if $\{L_1,L_2,L_3\}$ satisfies 
$CC_{BGSA}$. 
\end{theorem}
%
%
%
\paragraph{\textbf{n-GSA}.}
%
%
After addressing the binary case, we move to defining necessary and sufficient conditions for $n$-GSA.
A family of LISs $\{Itp_{L_1},\ldots,Itp_{L_{n+1}}\}$ has $n$-GSA if, for any $\Phi=\{\phi_1,\ldots,\phi_{n+1}\}$,
${I_{\Phi_1,L_1} \wedge \cdots \wedge I_{\phi_{n},L_{n}} \implies I_{\phi_1 \ldots \phi_n, L_{n+1}}}$, provided $\Phi$ is inconsistent.
As we defined a set of labeling constraints for BGSA, we now introduce \emph{n-GSA constraints} ($CC_{nGSA}$) on a family of labelings $\{L_1,\ldots,L_{n+1}\}$;
for every variable with labeling vector $(\alpha_{i_1},\ldots,\alpha_{i_{k+1}})$, $1 \leq k \leq n$,
letting  $ m = i_{k+1}$ if $i_{k+1} \neq n+1$,  $m = i_{k}$ otherwise:
{
\setlength{\abovedisplayskip}{5pt}
\setlength{\belowdisplayskip}{5pt}
 \begin{align*}
  & \text{(1)} \quad (\exists j \in \{i_1,\ldots,i_m\} \: \alpha_j=a) \implies
  (\forall h \in \{i_1,\ldots,i_m\} \: h \neq j \implies  \alpha_h=b)\\
 & \text{(2)} \quad \text{Moreover, if } i_{k+1} = n+1:
 \forall j \in  \{i_1,\ldots,i_{k}\}, \alpha_j \preceq \alpha_{i_{k+1}}
 \end{align*}
}
That is, if a variable is not shared with $\phi_{n+1}$, then, if
one of the labels is $a$, all the others must be $b$; if the variable
is shared with $\phi_{n+1}$, condition $(1)$ still holds for
$(\alpha_{i_1},\ldots,\alpha_{i_{k-1}})$, and all these
labels must be stronger or equal than $\alpha_{i_{k+1}}= \alpha_{{n+1}}$.
We  can prove that these constraints are necessary and sufficient for a family of LIS to have $n$-GSA:
%
%
%
\begin{theorem}
\label{theo:ngsa_nec_suff}
A family $\mF = \{Itp_{L_1},\ldots,Itp_{L_{n+1}}\}$ has $n$-GSA if and only if $\{L_1,\ldots,L_{n+1}\}$ satisfies 
$CC_{nGSA}$. 
\end{theorem}
 In \cite{RSS12} (see Setting~1) it is proved that $n$-SA holds for any
family of LISs stronger than Pudl\'ak. Theorem~\ref{theo:ngsa_nec_suff} is
strictly more general, since it allows for tuples of labels (e.g.,
$(\alpha_1,\alpha_2)=(a,b)$ or $(\delta_1,\delta_3,\delta_2) =
(a,b,b)$) that were not considered in~\cite{RSS12}.
The constraints for $n$-SA follow as a special case of $CC_{nGSA}$:
\begin{corollary}
\label{theo:nsa_nec_suff}
A family $\mF = \{Itp_{L_1},\ldots,Itp_{L_{n}}\}$ has $n$-SA if and only if $\{L_1,\ldots,L_{n}\}$ satisfies 
the following constraints: for every variable with labeling vector $(\alpha_{i_1},\ldots,\alpha_{i_{k}})$, for ${2\leq k\leq n}$: 
  ${(\exists j \in \{i_1,\ldots,i_k\} \, \alpha_j=a) \implies (\forall h \in \{i_1,\ldots,i_k\} \, h \neq j \implies  \alpha_h=b)}$.
\end{corollary}
Moreover, a family that has $(n+1)$-SA also has $n$-GSA if the last member of the family is Pudl\'ak's system. In fact,
from Proposition~\ref{prop:sa_to_bgsa} and 
  Pudl\'ak's system being symmetric (as shown in~\cite{H95}), it follows that
  \emph{if a family $\{Itp_{L_1},\ldots,Itp_{L_{n}},Itp_{P}\}$ has $(n+1)$-SA, then it has $n$-GSA}.




After investigating $n$-GSA and $n$-SA, we address two
questions which were left open in \S\ref{sec:interp_family}: do $n$-SA and $n$-PI 
imply $n$-STI? 
Is the requirement of additional interpolation systems
   necessary to obtain $T$-TI from $n$-GSA?
We show here that $n$-SA and $n$-PI do not necessarily imply $n$-STI, and that,
for LISs, $n$-GSA and $T$-TI are equivalent.

\vspace*{2mm}
\paragraph{\textbf{n-STI}.}
Theorem~\ref{lem:st_to_pi_sa} shows that
if a family has $n$-STI, then it has both $n$-SA and
$n$-PI. We prove that the converse is not necessarily true.
First, it is not difficult to show that any family  $\{Itp_{L_0},Itp_{L_1},Itp_{L_2}\}$ has $2$-PI (Proposition~\ref{prop:2pi} in the appendix);
a second result is that:
\setcounter{lemma}{4}
\vspace{-1mm}\begin{lemma}
  There exists a family $\{Itp_{L_0},Itp_{L_1},Itp_{L_2}\}$ that has
  $2$-PI and a family $\{Itp_{L'_1},Itp_{L'_2}\}$ that has $2$-SA, but the family 
  $\{Itp_{L_0},Itp_{L_1},Itp_{L_2},Itp_{L'_1},Itp_{L'_2}\}$ does not
  have $2$-STI.
\end{lemma}
\vspace{-1mm}
We obtain the main result applying the STI sub-family property (Theorem~\ref{lem:sti_sub}):
\vspace{-1mm}
\begin{theorem}
\label{thm:sa_pi_not_sti}
  There exists a family $\{Itp_{S_0},\ldots,Itp_{S_n}\}$ that has
  $n$-PI, and a family $\{Itp_{T_1},\ldots,Itp_{T_n}\}$ that has
  $n$-SA, but the family
  $\{Itp_{S_0},\ldots,Itp_{S_n}\} \cup$\lb $\{Itp_{T_1},\ldots,Itp_{T_n}\}$ does not
  have $n$-STI.
\end{theorem}

\paragraph{\textbf{T-TI}.} The last collective to be studied is $T$-TI.
Theorem~\ref{lem:gsa_to_ti} shows how $T$-TI can be obtained by multiple applications of GSA at the level
of each parent and its children, provided that we can find an appropriate labeling to generate an
interpolant for the parent. We prove here that, in the case of LISs, this requirement is not needed, and derive
explicit constraints on labelings for $T$-TI.

Let us define \emph{$n$-GSA strengthening} any property derived from $n$-GSA by not abstracting any of the subformulae $\phi_i$, for example
$I_{\phi_{1},L_{1}} \wedge \ldots \wedge I_{\phi_{n-1},L_{n-1}} 
\wedge \phi_{n} \implies I_{\phi_{1}\ldots \phi_n,L_{n+1}}$; it can be proved that:
\vspace{-1mm}\begin{lemma}
\label{lem:ngsa_str}
The set of labeling constraints of any $n$-GSA strengthening is a subset of constraints of $n$-GSA.
\end{lemma}
From Theorem~\ref{lem:gsa_to_ti} and Lemma~\ref{lem:ngsa_str}, it follows that:
\vspace{-1mm}\begin{lemma}
Given a  tree $T = (V,E)$ a family $\{Itp_{S_i}\}_{i \in V}$ has $T$-TI
 if, for every parent $i_{k+1}$ and its children $i_1,\ldots,i_{k}$, the family of labelings of the $(k+1)$-GSA
 strengthening obtained by non abstracting the parent satisfies the correspondent subset of $(k+1)$-GSA constraints. 
\end{lemma}
Note that, in contrast to Theorem~\ref{lem:gsa_to_ti}, in the case
of LISs we do not need to ensure the existence of an additional set of
interpolation systems to abstract the parents. The symmetry between
the necessary and sufficient conditions given by
Theorem~\ref{lem:gsa_to_ti} and Theorem~\ref{lem:ti_to_gsa} is restored,
and we establish:
\vspace{-1mm}\begin{theorem}
  Given a tree $T = (V,E)$ a family $\{Itp_{S_i}\}_{i \in V}$ has
  $T$-TI if and only if for every parent $i_{k+1}$ and its children
  $i_1,\ldots,i_{k}$, the family of labelings of the $(k+1)$-GSA
  strengthening obtained by non abstracting the parent satisfies the
  correspondent subset of $(k+1)$-GSA constraints.
\end{theorem}
Alternatively, in the case of LISs, the additional interpolation
systems can be constructed explicitly:
\vspace{-1mm}\begin{theorem}
\label{thm:ngsa_str_ext}
Any $\mF=\{Itp_{L_{i_1}},\ldots,Itp_{L_{i_k}},Itp_{L_{n+1}}\}$
s.t. $k< n$ that has an $n$-GSA strengthening property can be extended
to a family that has $n$-GSA.
\end{theorem}


\paragraph{\textbf{Collectives of Single LISs}.}
\label{sec:interp_sing_lab}
In the following, we highlight the fundamental results in the context of single LISs,
which represent the most common application of the framework of D'Silva et al. to SAT-based model checking.

First, importantly for practical applications, any LIS satisfies PI:
\vspace{-0.4mm}
\begin{theorem}
\label{theo:lis_pi}
PI holds for all single LISs.
\end{theorem}
\vspace{-0.4mm}

Second, recall that in \S\ref{sec:interp_sing} we proved that BGSA, STI,
TI, GSA are equivalent for single interpolation systems, and that SA
$\rightarrow$ BGSA for  symmetric ones.  We now show that for a single
LIS, SA is equivalent to BGSA and that PI is not.
\vspace{-2mm}
\begin{theorem}
\label{theo:lis_sa_bgsa}
If a LIS has SA, then it has BGSA.
\end{theorem}
\vspace{-3mm}
\begin{proof}
  We show that, for any $L$, the labeling constraints of SA imply
  those of BGSA. Refer to Table~\ref{tab:bgsa}, Table~\ref{tab:3sa},
  Theorem~\ref{theo:ngsa_nec_suff} and
  Corollary~\ref{theo:nsa_nec_suff}. In case of a family
  $\{L_1,L_2,L_3\}$, the constraints for $3$-SA are:
{
\setlength{\abovedisplayskip}{5pt}
\setlength{\belowdisplayskip}{5pt}
\begin{align*}
(\alpha_1,\alpha_2), (\beta_2,\beta_3), (\gamma_1,\gamma_3)
  \preceq \{(ab,ab),(b,a),(a,b)\}\\ 
(\delta_1,\delta_2,\delta_3)\preceq
  \{(ab,ab,ab),(a,b,b),(b,a,b),(b,b,a)\}
\end{align*}
}
When $L_1=L_2=L_3$, they simplify to 
$\alpha,\beta,\gamma,\delta \in \{ab,b\}$; this means that, in case
of a single LIS, only Pudl\'ak's or stronger systems are allowed.
In case of a family $\{L_1,L_2,L_3\}$, the constraints for BGSA are:
{
\setlength{\abovedisplayskip}{5pt}
\setlength{\belowdisplayskip}{5pt}
\[(\alpha_1,\alpha_2), (\delta_1,\delta_2) \preceq
  \{(ab,ab),(b,a),(a,b)\}, 
\beta_2 \preceq \beta_3, \gamma_1 \preceq \gamma_3, \delta_1 \preceq \delta_3, \delta_2 \preceq \delta_3
\]
}
When $L_1=L_2=L_3$, they simplify to $\alpha,\delta \in \{ab,b\}$;
clearly, the constraints for $3$-SA imply those for BGSA, but not vice
versa.
\end{proof}
%
%
\vspace{-2mm}
Finally, Theorem~\ref{theo:lis_pi} and Theorem~\ref{theo:lis_sa_bgsa}  yield:
\begin{theorem}
\label{theo:pi_not_bgsa}
The system $Itp_{M'}$ has PI but does not have BGSA.
\end{theorem}
\vspace{-2mm}
\begin{proof}
  From the \emph{proof} of Theorem~\ref{theo:lis_sa_bgsa}: a LIS has
  the BGSA property iff it is stronger or equal than Pudl\'ak's
  system. $Itp_{M'}$ is strictly weaker than $Itp_{P}$. Thus, it does
  not have BGSA.
\end{proof}

\vspace{-2mm}
Note that the necessary and sufficient conditions for LISs to support
each of the collectives simplify implementing procedures with a
given property, or, more importantly from a practical perspective,
determine which implementation supports which property.


\vspace{-3mm}
\section{Implementation}
\vspace{-2.5mm}
\label{sec:implementation}

We developed an interpolating prover, PeRIPLO\footnote{PeRIPLO is available at http://verify.inf.usi.ch/periplo.html},
which implements the proposed framework. PeRIPLO is, to the best of our knowledge, the first SAT-solver built on MiniSAT 2.2.0 that
realizes the Labeled Interpolation Systems of \cite{DKPW10} and allows to perform interpolation, path interpolation, generalized simultaneous
abstraction, state-transition interpolation and tree interpolation;  it also offers proof logging and manipulation routines. 
The tool has been integrated within the FunFrog and eVolCheck verification frameworks,
which make use of its solving and interpolation features for SAT-based model checking.
In theory, using different
partitions of the same formula and  different
labelings with each partition does not change the
algorithmic complexity of LISs (see appendix~\ref{sec:app:3}). In our experience, there is no overhead in practice as well.
\vspace{-3mm}
\section{Conclusions}
\vspace{-2.5mm}
\label{sec:concl}

Craig interpolation is a widely used approach in abstraction-based
model checking. 
%
%
%
This paper conducts a systematic investigation of
the most common interpolation properties exploited in verification,
focusing on the constraints they pose on propositional interpolation systems used 
in SAT-based model checking.  

The paper makes the following contributions. 
It
systematizes and unifies various properties imposed on interpolation
by existing verification approaches and proves that for families of
interpolation systems the properties form a hierarchy, whereas for a
single system all properties except path interpolation and
simultaneous abstraction are in fact equivalent. 
Additionally, it
defines and proves both sufficient and necessary conditions for a
family of Labeled Interpolation Systems.  In particular, it
demonstrates that in case of a single system path interpolation is
common to all LISs, while simultaneous abstraction is as strong as all
other more complex properties.
Extending our framework to address interpolation in first order theories is an
interesting open problem, and is part of our future work.

\vspace{-4mm}
\bibliographystyle{abbrv}
\bibliography{biblio}

\newpage
\appendix

\section{Properties of Sub-families}

\setcounter{theorem}{1}
\begin{theorem}
A family $\{Itp_{S_0},\ldots,Itp_{S_n},Itp_{T_1},\ldots,Itp_{T_n}\}$  has $n$-STI iff for all $k\leq n$ the subfamily
$\{Itp_{S_{0}},\ldots,Itp_{S_{k}}\}$ $\cup$ $\{Itp_{T_{1}},\ldots,Itp_{T_{k}}\}$  has $k$-STI.
\end{theorem}
\begin{proof}
$\rightarrow)$ Assume an inconsistent $\Phi \triangleq \{\phi_1,\ldots,\phi_k\}$. We can extend it to a 
$\Phi' \triangleq \{\phi'_1,\ldots,\phi'_n\}$ such that $\phi'_i \equiv \phi_i$,
by adding
$n-k$ empty formulae $\top$. If $\mF$ has the $n$-STI property, for $0\leq j\leq k-1$
$$I_{\phi_1\ldots \phi_j,S_j} \wedge I_{\phi_{j+1},T_{j+1}} \rightarrow I_{\phi_1\ldots \phi_{j+1},S_{j+1}}$$
$\leftarrow)$ Follows from $k=n$.
\end{proof}

\setcounter{theorem}{16}
\begin{theorem}
\label{lem:gsa_sub}
A family $\mF = \{Itp_{S_1},\ldots,Itp_{S_{n+1}}\}$ has 
$n$-GSA iff for all $k\leq n$ all the subfamilies
$\{Itp_{S_{i_1}},\ldots,Itp_{S_{i_{k+1}}}\}$ have \mbox{$k$-GSA}.
\end{theorem}
\begin{proof}
  ($\rightarrow$) Let $n$ be a natural number. Take any inconsistent
  $\Phi = \{\phi_1,\ldots,\phi_{k+1}\}$ such that $k \leq
  n$. Let $\{i_1, \ldots, i_{k+1}\}$ be a subset of
    $\{1,\ldots,n+1\}$. Extend $\Phi$ to a ${\Phi' =
    \{\phi'_1,\ldots,\phi'_{n+1}\}}$ by adding $(n-k)$ copies of $\top$,
    so that $\phi'_{i_1} = \phi_1, \ldots,
    \phi'_{i_{k}} = \phi_{k},$ $\phi'_{i_{k+1}} = \phi_{n+1}$.
    Since $\mF$ has $n$-GSA:
    \[
     \bigwedge_{j=1}^{n} I_{\phi'_j,S_j} \implies I_{\phi'_1 \ldots \phi'_{n},S_{n+1}} 
    \]
    and, since $\phi'_j = \top$ for $j \not\in \{i_1, \ldots, i_k\}$:
    \[
    \bigwedge_{j\in \{i_1\ldots i_{k}\}} I_{\phi_j,S_j} \implies
    I_{\phi_{i_1 \ldots i_{k}},S_{i_{k+1}}}
    \]
$(\leftarrow)$ Follows from $k=n$.
\end{proof}

It is easy to see that the technique used in the proof of
Theorem~\ref{lem:gsa_sub}, i.e., extending an unsatisfiable formula with
$\top$ conjuncts, applies to the other properties as well.

\begin{theorem}
\label{lem:sa_sub}
A family $\{Itp_{S_1},\ldots,Itp_{S_n}\}$ has $n$-SA iff
for all $k\leq n$ all the subfamilies
$\{Itp_{S_{i_1}},\ldots,Itp_{S_{i_k}}\}$ have $k$-SA.
\end{theorem}
\begin{proof}
The proof works as in Theorem~\ref{lem:gsa_sub}.
\end{proof}

\begin{theorem}
\label{lem:pi_sub}
A family  $\{Itp_{S_0},\ldots,Itp_{S_n}\}$ has $n$-PI iff for all $k\leq n$ the subfamily $\{Itp_{S_{0}},\ldots,Itp_{S_{k}}\}$  
has $k$-PI.
\end{theorem}
\begin{proof}
The proof works as in Theorem~\ref{lem:sti_sub}.
\end{proof}

\begin{theorem}
\label{lem:ti_sub}
  For a given tree $T = (V,E)$, a family $\{Itp_{S_i}\}_{i \in V}$ has $T$-TI
  iff for every subtree $T' = (V', E')$ of $T$, the family
  $\{Itp_{S_j}\}_{j \in V'}$ has $T'$-TI.
\end{theorem}
\begin{proof}
$\rightarrow)$. Assume an inconsistent $\Phi\triangleq \{\phi_{i_1},\ldots,\phi_{i_k}\}$ decorating $T'$. We can 
extend $\Phi$ with $|V'|-|V|$ empty formulae $\top$ to $\Phi'\triangleq \{\phi'_1,\ldots,\phi'_n\}$ decorating $T$. 
If $\{Itp_{S_i}\}_{v_i \in V}$ has the $T$-TI property, for all $v'_i$ in $V$ and in particular for all $v_i$ in $V'$
$$\bigwedge_{(v_i,v_j) \in E'} I_{F_j,S_j} \wedge \phi_i \rightarrow I_{F_i,S_i} $$
$\leftarrow)$. Follows from $T'\equiv T$.
\end{proof}

\section{Other Proofs}

\begin{proposition}
  \label{lem:single_sa_to_bgsa}
  SA implies BGSA in symmetric interpolation systems.
\end{proposition}
\begin{proof}
  Take any inconsistent $\Phi = \{\phi_1,\phi_2,\phi_3\}$.
  If an interpolation system has SA, then:
 \[I_{\phi_1} \wedge I_{\phi_2} \wedge I_{\phi_3} \implies \bot\]
 Equivalently,
 \[I_{\phi_1} \wedge I_{\phi_2} \implies  \n{I_{\phi_3}}\] \vspace{1mm}
 For a symmetric system, $\n{I_{\phi_3}} =  I_{\phi_1
   \phi_2}$.
\end{proof}

\begin{proposition}
\label{prop:sa_to_bgsa}
If a family $\mF = \{Itp_{S_1},\ldots,Itp_{S_{n+1}}\}$ has 
\mbox{$(n+1)$-SA}  and $Itp_{S_{n+1}}$ is symmetric, then $\mF$ has
$n$-GSA.
\end{proposition}
\begin{proof}
  Take any inconsistent $\Phi =
  \{\phi_1,\ldots,\phi_n\}$. Since $\mF$ has $(n+1)$-SA, then
  $I_{\phi_1,S_1} \wedge \cdots \wedge I_{\phi_{n+1},S_{n+1}} \implies
  \bot$.  Assuming $Itp_{S_{n+1}}$ is symmetric, 
  $\n{I_{\phi_{n+1},S_{n+1}}} = I_{\phi_1,\ldots,\phi_n,S_{n+1}}$ and
  the thesis is proved.
\end{proof}

\setcounter{theorem}{2}
\begin{theorem}
If a family $\mF = 
\{Itp_{S_0},\ldots,Itp_{S_n},Itp_{T_1},\ldots,Itp_{T_n}\}$ has
$n$-STI then (1) $\{Itp_{S_0},\ldots,Itp_{S_n}\}$ has
$n$-PI and (2) $\{Itp_{T_1},\ldots,Itp_{T_n}\}$ has 
$n$-SA.
\end{theorem}
\begin{proof}
  $(1)$ It follows from  $\phi_i \implies I_{\phi_i,S_i}$ for every $i$.\\

\nin  $(2)$. Take any inconsistent $\Phi =
  \{\phi_1,\ldots,\phi_{n}\}$.  If $\mF$ has $n$-STI,
  then, for $0\leq i \leq n-1$:
\[I_{\phi_1\cdots \phi_i,S_i} \wedge I_{\phi_{i+1},T_{i+1}} \implies I_{\phi_1\cdots \phi_{i+1},S_{i+1}}\]
Since $I_{\phi_1 \cdots \phi_n} = \bot$, we get $I_{\phi_1,T_1} \wedge \cdots \wedge I_{\phi_n,T_n} \implies \bot$.
\end{proof}

\begin{theorem}
A family $\mF = \{Itp_{S_0},\ldots,Itp_{S_n},Itp_{T_1},\ldots,Itp_{T_n}\}$
has $n$-STI iff
$\{Itp_{S_i},Itp_{T_{i+1}},Itp_{S_{i+1}}\}$ has BGSA for
all  $0\leq i \leq n-1$.
\end{theorem}
\begin{proof}
  $(\rightarrow)$.  Take any inconsistent
  $\Phi = \{\phi_1,\phi_2,\phi_3\}$.  For $0\leq i \leq
  n-1$, extend $\Phi$ to a $\Phi'=
  \{\phi'_1,\ldots,\phi'_n\}$ by adding $(n-3)$ copies of $\top$, so
  that $\phi'_i = \phi_1$, $\phi'_{i+1} = \phi_2$,
  $\phi'_{i+2} =  \phi_3$. Since $\mF$ has $n$-STI:
\[
I_{\phi'_1\cdots \phi'_i,S_i} \wedge I_{\phi'_{i+1},T_{i+1}}
\implies I_{\phi'_1\cdots \phi'_{i+1},S_{i+1}}
\]
Hence, by construction:
\[
I_{\phi_1,S_i} \wedge I_{\phi_2,T_{i+1}} \implies I_{\phi_1 \phi_2,S_{i+1}}
\]

\nin $(\leftarrow)$
Take any inconsistent $\Phi = \{\phi_1,\ldots,\phi_{n}\}$.
Since $\{Itp_{S_i},Itp_{T_{i+1}},Itp_{S_{i+1}}\}$ has BGSA,
it follows that for $\{\phi'_1,\phi'_2,\phi'_3\}$, where $\phi'_1=\phi_1 \wedge \cdots \wedge \phi_i$, $\phi'_2=\phi_{i+1}$, $\phi'_3=\phi_{i+2}
\wedge \cdots \wedge \phi_n$:
\[I_{\phi'_1,S_i} \wedge I_{\phi'_2,T_{i+1}} \implies I_{\phi'_1 \phi'_2,S_{i+1}} \]
Hence, by construction:
\[I_{\phi_1\ldots \phi_i,S_i} \wedge I_{\phi_{i+1},T_{i+1}} \implies
I_{\phi_1\ldots \phi_{i+1},S_{i+1}}\]
%
\end{proof}

\begin{theorem}
  Given a tree $T = (V,E)$ if a family ${\mF = \{Itp_{S_i}\}_{i \in V}}$ has $T$-TI, then, for every
  parent $i_{k+1}$ and its children $i_1,\ldots,i_{k}$:
  \begin{enumerate}
  \item If $i_{k+1}$ is the root,
    $\{Itp_{S_{i_1}},\ldots,Itp_{S_{i_{k}}}\}$ has $k$-SA.
  \item Otherwise,
    $\{Itp_{S_{i_1}},\ldots,Itp_{S_{i_{k}}},Itp_{S_{i_{k+1}}}\}$ has
    $k$-GSA.
  \end{enumerate}
\end{theorem}
\begin{proof}
  Take any inconsistent
  $\Phi = \{\phi_{i_1},\ldots,\phi_{i_{k+1}}\}$. Consider a
  parent $i_{k+1}$ and its children $i_1,\ldots,i_{k}$.  If
  $i_{k+1}$ is not the root, extend $\Phi$ to a $\Phi'$ in such a way that: the
  children are decorated with $\phi_{i_1},\ldots,\phi_{i_{k}}$, all
  their descendants and $i_{k+1}$ with $\top$, all the nodes
  external to the subtree rooted in $i_{k+1}$ with $\phi_{n+1}$.
  Since $\mF$ has $T$-TI, then at node $i_{k+1}$:
\[\bigwedge_{(i_{k+1},j) \in E} I_{F_j,S_j} \wedge \phi_{i_{k+1}} \implies I_{F_{i_{k+1}},S_{i_{k+1}}} \]
that is:
\[\bigwedge_{i \in \{i_1 \ldots i_k\}} I_{\phi_i,S_i} \wedge \top \implies I_{\phi_{i_1}\cdots \phi_{i_k},S_{k+1}}\]
If $i_{k+1}$ is the root, the proof simply ignores the presence of $\phi_{i_{k+1}}$ and $S_{i_{k+1}}$.
\end{proof}

\begin{theorem}
 Given a  tree $T = (V,E)$, a family $\mF = \{Itp_{S_i}\}_{i \in V}$ has $T$-TI
  if, for every node $i_{k+1}$ and its children $i_1,\ldots,i_{k}$, there exists $T_{i_{k+1}}$ such that:
  \begin{enumerate}
  \item If $i_{k+1}$ is the root,
    $\{Itp_{S_{i_1}},\ldots,Itp_{S_{i_{k}}},Itp_{T_{i_{k+1}}}\}$ has ${(k+1)}$-SA.
  \item Otherwise, 
   $\{Itp_{S_{i_1}},\ldots,Itp_{T_{i_{k+1}}},Itp_{S_{i_{k+1}}}\}$ has \mbox{${(k+1)}$-GSA}.
  \end{enumerate}
\end{theorem}
\begin{proof}
  Take any inconsistent $\Phi = \{\phi_{1},\ldots,\phi_{{n}}\}$. 
  Consider a parent $i_{k+1}$ different from the root and its
  children $i_1,\ldots,i_{k}$.\\
  If  $\{Itp_{S_{i_1}},\ldots,Itp_{T_{i_{k+1}}},Itp_{S_{i_{k+1}}}\}$
  has \mbox{$k$-GSA}, for 
  $\{F_{i_1},\ldots,F_{i_k},\phi_{i_{k+1}},\Phi \setminus (\bigcup
  F_{i_j} \cup \{\phi_{i_{k+1}}\})\}$:
  \[\bigwedge_{i\in \{i_1 \ldots i_{k}\}} I_{F_i,S_i} \wedge
  I_{\phi_{i_{k+1}},T_{i_{k+1}}} \implies I_{F_{i_{k+1}},S_{i_{k+1}}}\]
  The thesis follows since $\phi_{i_{k+1}} \implies I_{\phi_{i_{k+1}},T_{i_{k+1}}}$.
  If $i_{k+1}$ is the root, $I_{F_{i_{k+1}},S_{i_{k+1}}}=\bot$ and $S_{i_{k+1}}$ is superfluous.
\end{proof}

\begin{figure}[ht]
\vspace{-0.5cm}
  \centering
  \begin{minipage}{0.45\textwidth}
 \begin{tikzpicture}[node distance=1cm]
  \node (2) {$0$};
  \node (1) [right=1cm of 2] {$\underset{\phi_{n+1}}{n+1}$};
     \node[draw=none] (6) [left of=2] {};
   \node[draw=none] (7) [right of=2] {};
  \node (3) [below of=6] {$\underset{\phi_1}{1}$};
  \node (4) [below of=7] {$\underset{\phi_n}{n}$};
  \node (5) [below of=2] {$\cdots$};

  \path[->,every node/.style={font=\sffamily\small}]
    (1) edge node {} (2)
    (2) edge node {} (3)
	edge node  {} (4);    
\end{tikzpicture}
\caption{$T^n_{GSA}$.}
\label{fig:gsa-tree}
\end{minipage}
\begin{minipage}{0.45\textwidth}
\vspace{-0.8cm}
 \begin{tikzpicture}[node distance=1cm]
   \node (5)  {$n+1$};
  \node (1) [below of=5] {$\underset{\phi_1}{1}$};
  \node (2) [right=1cm of 1] {$\underset{\phi_2}{2}$};
  \node (3) [right=1cm of 2] {$\cdots$};
  \node (4) [right=1cm of 3] {$\underset{\phi_n}{n}$};
  \node (6) [above of=2] {$n+2$};
  \node (7) [above of=3] {$\cdots$};
  \node (8) [above of=4] {$2n$};
  \node (9) [above of=6] {$\overset{}{}$};

  \path[->,every node/.style={font=\sffamily\small}]
    (8) edge node {} (7)
	edge node {}  (4)
    (7) edge node {} (6)
    (6) edge node {} (5)
	edge node {} (2) 
    (5) edge node {} (1);
\end{tikzpicture}
  \caption{$T^n_{STI}$.}
  \label{fig:sti-tree}
\end{minipage}
\vspace{-0.5cm}
\end{figure}

\begin{theorem}
  If a family $\mF = \{Itp_{S_{n+1}}, Itp_{S_1}, \ldots,
  Itp_{S_{n+1}}\}$ has $T_{GSA}^n$-TI, then $\{Itp_{S_1},
  \ldots, Itp_{S_{n+1}}\}$ has $n$-GSA.
\end{theorem}
\begin{proof}
Let $T^n_{GSA} = (V, E)$ be the tree shown in
Fig.~\ref{fig:gsa-tree}, where
$V = \{0, \ldots, n+1\}$ and $E = \{ (0,i) \mid 1 \leq i \leq n\} \cup \{(n+1,0)\} $.

  Take any inconsistent $\Phi =
  \{\phi_1,\ldots,\phi_{n+1}\}$. We decorate node $0$
  with $\top$, all other nodes $i$ with $\phi_i$, for $1\leq i \leq n+1$.  Since
  $\mF$ has $T$-TI, then at node $0$:
\[\bigwedge_{(0,j) \in E} I_{F_j,S_j} \wedge \top \implies I_{F_{0},S_{n+1}} \]
Hence, by construction:
\[\bigwedge_{i=1}^{n} I_{\phi_i,S_i} \implies I_{\phi_1\ldots \phi_n,S_{n+1}}\]
\end{proof}

\begin{theorem}
  If a family $\mF =
  \{Itp_{S_0},\ldots,Itp_{S_n}\}\cup{}$ $\{Itp_{T_1},\ldots,Itp_{T_n}\}$ has 
  $T^n_{STI}$-TI, then it has $n$-STI.
\end{theorem}
\begin{proof}
Let $T^n_{STI} = (V, E)$ be the tree shown in Fig.~\ref{fig:sti-tree}, where 
$V = \{1,\ldots,2n\}$ and 
  $E = \{ (n+i, i) \mid 1\leq i \leq n \} \cup
  \{(n+i,n+i-1) \mid 1\leq i \leq n\}$.
  
  Take any inconsistent $\Phi =
  \{\phi_1,\ldots,\phi_{n}\}$. For ${1\leq i\leq n}$, we decorate $i$ with $\phi_i$, $n+i$ with $\top$;
  similarly we associate $i$ with $Itp_{T_i}$ and $n+i$ with $Itp_{S_i}$.
   Since $\mF$ has $T$-TI, then
  at every node $n+i+1$, for $0 \leq i \leq n-1$:
\[(I_{F_{n+i},S_i} \wedge  I_{F_{i+1},T_{i+1}}) \wedge \top \implies I_{F_{n+i+1},S_{i+1}}\]
Hence, by construction,
\[I_{\phi_1\ldots \phi_i,S_i} \wedge I_{\phi_{i+1},T_{i+1}} \implies I_{\phi_1\ldots \phi_{i+1},S_{i+1}}\]
\end{proof}

\setcounter{lemma}{0}
\begin{lemma}
If $\{L_1,L_2,L_3\}$ satisfies $CC_{BGSA}^*$, then $\{Itp_{L_1},Itp_{L_2},Itp_{L_3}\}$ has BGSA. 
\end{lemma}
\begin{proof}[by structural induction]
We remind here the \emph{restricted BGSA constraints} $CC_{BGSA}^*$: 
\[
(\alpha_1,\alpha_2), (\delta_1,\delta_2) \in \{(ab,ab),(b,a),(a,b)\}, 
\beta_2 = \beta_3, \gamma_1 = \gamma_3, \delta_3 = \max\{\delta_1,\delta_2\}
 \]
The reader can verify that the conditions on the $\delta_i$ are equivalent to 
$(\delta_1, \delta_2, \delta_3) \in \{ (ab,ab,ab), (b,a,a), (a,b,a) \}$.

We show that, given a refutation of $\Phi$, for any clause $C$ in the refutation 
the partial interpolants satisfy $I_{\phi_1,L_1}(C) \wedge I_{\phi_2,L_2}(C) \implies  I_{\phi_1 \phi_2,L_3}(C)$,
that is $I_{\phi_1,L_1}(C) \wedge I_{\phi_2,L_2}(C) \wedge \n{I_{\phi_1 \phi_2,L_3}(C)} \implies  \bot$.

For simplicity, we write $I_1$, $I_2$, $I_3$ to refer to the three partial interpolants for $C$ and, if $C$ has antecedents, we  denote
their partial interpolants with $I^+_1$, $I^+_2$, $I^+_3$ and $I^-_1$, $I^-_2$, $I^-_3$.

\bfstart{Base case} (leaf). Case splitting on $C$ (refer to Table~\ref{tab:bgsa}): 
\begin{description}
\item[$C \in \phi_1$]: $\quad$ 
$I_1 = C\!\!\downharpoonright_{1,b}$ $\quad$ $I_2 = \n{C\!\!\downharpoonright_{2,a}}$ $\quad$ $\n{I_3} = \n{C\!\!\downharpoonright_{3,b}}$ 
\item[$C \in \phi_2\,$]: $\quad$
$I_1 = \n{C\!\!\downharpoonright_{1,a}}$ $\quad$ $I_2 = C\!\!\downharpoonright_{2,b}$ $\quad$ $\n{I_3} = \n{C\!\!\downharpoonright_{3,b}}$ 
\item[$C \in \phi_3$]: $\quad$
$I_1 = \n{C\!\!\downharpoonright_{1,a}}$ $\quad$ $I_2 = \n{C\!\!\downharpoonright_{2,a}}$ $\quad$ $\n{I_3} = C\!\!\downharpoonright_{3,a}$ 
\end{description}
The goal is to show that in each case ${I_1 \wedge I_2 \wedge \n{I_3} \implies \bot}$.
Representing $C$ by grouping variables into the different partitions, with overbraces to show the label assigned to each
variable, we have:
\begin{description}
\item[$C \in \phi_1$]:\\
$C\!\!\downharpoonright_{1,b} = \overbrace{C_{\phi_1}\!\!\downharpoonright_{b}}^{a} \vee \overbrace{C_{\phi_1 \phi_2}\!\!\downharpoonright_{b}}^{\alpha_1} 
\vee \overbrace{C_{\phi_1\phi_3}\!\!\downharpoonright_{b}}^{\gamma_1} \vee \overbrace{C_{\phi_1 \phi_2 \phi_3}\!\!\downharpoonright_{b}}^{\delta_1}$\\
$\n{C\!\!\downharpoonright_{2,a}} = \overbrace{\n{C_{\phi_1}\!\!\downharpoonright_{a}}}^{b} \wedge \overbrace{\n{C_{\phi_1 \phi_2}\!\!\downharpoonright_{a}}}^{\alpha_2} 
\wedge \overbrace{\n{C_{\phi_1\phi_3}\!\!\downharpoonright_{a}}}^{b} \wedge \overbrace{\n{C_{\phi_1 \phi_2 \phi_3}\!\!\downharpoonright_{a}}}^{\delta_2}$\\
$\n{C\!\!\downharpoonright_{3,b}} = \overbrace{\n{C_{\phi_1}\!\!\downharpoonright_{b}}}^{a} \wedge \overbrace{\n{C_{\phi_1 \phi_2}\!\!\downharpoonright_{b}}}^{a} 
\wedge \overbrace{\n{C_{\phi_1\phi_3}\!\!\downharpoonright_{b}}}^{\gamma_3} \wedge \overbrace{\n{C_{\phi_1 \phi_2 \phi_3}\!\!\downharpoonright_{b}}}^{\delta_3}$\\
\item[$C \in \phi_2\,$]:\\
$\n{C\!\!\downharpoonright_{1,a}} = \overbrace{\n{C_{\phi_2}\!\!\downharpoonright_{a}}}^{b} \wedge \overbrace{\n{C_{\phi_1\phi_2}\!\!\downharpoonright_{a}}}^{\alpha_1} 
\wedge \overbrace{\n{C_{\phi_2\phi_3}\!\!\downharpoonright_{a}}}^{b} \wedge \overbrace{\n{C_{\phi_1\phi_2\phi_3}\!\!\downharpoonright_{a}}}^{\delta_1}$\\
$C\!\!\downharpoonright_{2,b} = \overbrace{C_{\phi_2}\!\!\downharpoonright_{b}}^{a} \vee \overbrace{C_{\phi_1\phi_2}\!\!\downharpoonright_{b}}^{\alpha_2} 
\vee \overbrace{C_{\phi_2\phi_3}\!\!\downharpoonright_{b}}^{\beta_2} \vee \overbrace{C_{\phi_1\phi_2\phi_3}\!\!\downharpoonright_{b}}^{\delta_2}$\\
$\n{C\!\!\downharpoonright_{3,b}} = \overbrace{\n{C_{\phi_2}\!\!\downharpoonright_{b}}}^{a} \wedge \overbrace{\n{C_{\phi_1\phi_2}\!\!\downharpoonright_{b}}}^{a} 
\wedge \overbrace{\n{C_{\phi_2\phi_3}\!\!\downharpoonright_{b}}}^{\beta_3} \wedge \overbrace{\n{C_{\phi_1\phi_2\phi_3}\!\!\downharpoonright_{b}}}^{\delta_3}$\\
\item[$C \in \phi_3$]:\\ 
$\n{C\!\!\downharpoonright_{1,a}} = \overbrace{\n{C_{\phi_3}\!\!\downharpoonright_{a}}}^{b} \wedge \overbrace{\n{C_{\phi_2\phi_3}\!\!\downharpoonright_{a}}}^{b} 
\wedge \overbrace{\n{C_{\phi_1\phi_3}\!\!\downharpoonright_{a}}}^{\gamma_1} \wedge \overbrace{\n{C_{\phi_1\phi_2\phi_3}\!\!\downharpoonright_{a}}}^{\delta_1}$\\
$\n{C\!\!\downharpoonright_{2,a}} = \overbrace{\n{C_{\phi_3}\!\!\downharpoonright_{a}}}^{b} \wedge \overbrace{\n{C_{\phi_2\phi_3}\!\!\downharpoonright_{a}}}^{\beta_2} 
\wedge \overbrace{\n{C_{\phi_1\phi_3}\!\!\downharpoonright_{a}}}^{b} \wedge \overbrace{\n{C_{\phi_1\phi_2\phi_3}\!\!\downharpoonright_{a}}}^{\delta_2}$\\
$C\!\!\downharpoonright_{3,a} = \overbrace{C_{\phi_3}\!\!\downharpoonright_{a}}^{b} \vee \overbrace{C_{\phi_2\phi_3}\!\!\downharpoonright_{a}}^{\beta_3} 
\vee \overbrace{C_{\phi_1\phi_3}\!\!\downharpoonright_{a}}^{\gamma_3} \vee \overbrace{C_{\phi_1\phi_2\phi_3}\!\!\downharpoonright_{a}}^{\delta_3}$\\
\end{description}
We can carry out some simplifications, due to the equality constraints in $CC_{BGSA}^*$ and the fact that 
variables with label $a$ restricted w.r.t. $b$ (and vice versa) are removed, leading (with the help of the resolution rule) to the constraints:
$$(\overbrace{C_{\phi_1 \phi_2}\!\!\downharpoonright_{b}}^{\alpha_1} 
\vee \overbrace{C_{\phi_1 \phi_2 \phi_3}\!\!\downharpoonright_{b}}^{\delta_1}) \wedge 
\overbrace{\n{C_{\phi_1 \phi_2}\!\!\downharpoonright_{a}}}^{\alpha_2} 
\wedge \overbrace{\n{C_{\phi_1 \phi_2 \phi_3}\!\!\downharpoonright_{a}}}^{\delta_2} \wedge
\overbrace{\n{C_{\phi_1 \phi_2 \phi_3}\!\!\downharpoonright_{b}}}^{\delta_3} \implies \bot$$
$$\overbrace{\n{C_{\phi_1\phi_2}\!\!\downharpoonright_{a}}}^{\alpha_1} 
\wedge \overbrace{\n{C_{\phi_1\phi_2\phi_3}\!\!\downharpoonright_{a}}}^{\delta_1} \wedge
(\overbrace{C_{\phi_1\phi_2}\!\!\downharpoonright_{b}}^{\alpha_2} 
\vee \overbrace{C_{\phi_1\phi_2\phi_3}\!\!\downharpoonright_{b}}^{\delta_2}) \wedge 
\overbrace{\n{C_{\phi_1\phi_2\phi_3}\!\!\downharpoonright_{b}}}^{\delta_3} \implies \bot $$
$$\overbrace{\n{C_{\phi_1\phi_2\phi_3}\!\!\downharpoonright_{a}}}^{\delta_1} \wedge 
\overbrace{\n{C_{\phi_1\phi_2\phi_3}\!\!\downharpoonright_{a}}}^{\delta_2} \wedge 
\overbrace{C_{\phi_1\phi_2\phi_3}\!\!\downharpoonright_{a}}^{\delta_3} \implies \bot$$
Finally, the constraints on $(\alpha_1,\alpha_2)$ and $(\delta_1,\delta_2,\delta_3)$ guarantee that the remaining variables are simplified away,
proving the base case.

\bfstart{Inductive step} (inner node). The inductive hypothesis (i.h.) consists of  $I^+_1 \wedge I^+_2 \wedge I^+_3 \implies \bot$, 
$I^-_1 \wedge I^-_2  \wedge \n{I^-_3} \implies \bot$. 
We do a case splitting on the pivot $p$:

\bfstart{Case 1}  ($p$ in $\phi_1$).
\begin{flalign*}%
I_1  \wedge I_2 \wedge \n{I_3} & \iff \\
( I_1^+ \vee I_1^-)  \wedge (I_2^+ \wedge I_2^-)\wedge \n{( I_3^+ \vee I_3^-)}& \iff \\ 
(I_1^+ \vee I_1^-) \wedge I_2^+ \wedge I_2^- \wedge \n{I_3^+} \wedge \n{I_3^-}  & \implies\\
(I_1^+ \wedge I_2^+ \wedge \n{I_3^+} )  \vee  (I_1^- \wedge I_2^-  \wedge \n{I_3^-})& \implies^{\text{i.h.}}
\bot 
\end{flalign*}%
\bfstart{Case 2}  ($p$ in $\phi_2$).
\begin{flalign*}%
I_1  \wedge I_2 \wedge \n{I_3} & \iff \\
( I_1^+ \wedge I_1^-)  \wedge (I_2^+ \vee I_2^-)\wedge \n{( I_3^+ \vee I_3^-)}& \iff \\ 
I_1^+ \wedge I_1^- \wedge (I_2^+ \vee I_2^-) \wedge \n{I_3^+} \wedge \n{I_3^-}  & \implies\\
(I_1^+ \wedge I_2^+ \wedge \n{I_3^+} )  \vee  (I_1^- \wedge I_2^-  \wedge \n{I_3^-})& \implies^{\text{i.h.}}
\bot 
\end{flalign*}%
\bfstart{Case 3} ($p$ in $\phi_3$).
\begin{flalign*}%
I_1  \wedge I_2 \wedge \n{I_3}& \iff \\
( I_1^+ \wedge I_1^-) \wedge (I_2^+ \wedge I_2^-) \wedge \n{( I_3^+ \wedge I_3^-)}& \iff \\ 
I_1^+ \wedge I_1^-  \wedge I_2^+ \wedge I_2^- \wedge (\n{I_3^+} \vee \n{I_3^-})& \implies\\
(I_1^+  \wedge I_2^+ \wedge \n{I_3^+})  \vee  (I_1^- \wedge I_2^- \wedge \n{I_3^-} )& \implies^{\text{i.h.}}
\bot 
\end{flalign*}%
\bfstart{Case 4} ($p$ in $\phi_1 \phi_2$). If $(\alpha_1, \alpha_2)=(ab,ab)$:
\begin{adjustwidth}{-3em}{0em}
\begin{flalign*}
I_1  \wedge I_2 \wedge \n{I_3}& \iff \\
(I_1^+ \vee p) \wedge (I_1^- \vee \n{p})  \wedge (I_2^+ \vee p) \wedge (I_2^- \vee \n{p}) \wedge \n{(I_3^+ \vee I_3^-)}& \implies \\ 
(I_1^+ \vee p) \wedge (I_1^- \vee \n{p})  \wedge (I_2^+ \vee p) \wedge (I_2^- \vee \n{p}) \wedge (\n{I_3^+}\vee p) \wedge (\n{I_3^-}\vee \n{p}) & \implies \\ 
((I_1^+ \wedge I_2^+ \wedge \n{I_3^+}) \vee p) \wedge  ((I_1^- \wedge I_2^- \wedge \n{I_3^-}) \vee \n{p}) & \implies^{\text{resol}}\\
(I_1^+  \wedge I_2^+ \wedge \n{I_3^+})  \vee  (I_1^-  \wedge I_2^- \wedge \n{I_3^-})& \implies^{\text{i.h.}}
\bot 
\end{flalign*} 
\end{adjustwidth}
\bfstart{Case 5} ($p$ in $\phi_1 \phi_2 \phi_3$). If $(\delta_1, \delta_2, \delta_3)=(ab,ab,ab)$:
\begin{adjustwidth}{-3em}{0em}
\begin{flalign*}
I_1  \wedge I_2 \wedge \n{I_3}& \iff \\
(I_1^+ \vee p) \wedge (I_1^- \vee \n{p})  \wedge (I_2^+ \vee p) \wedge (I_2^- \vee \n{p}) \wedge \n{((I_3^+ \vee p) \wedge (I_3^- \vee \n{p}))}& \iff \\ 
(I_1^+ \vee p) \wedge (I_1^- \vee \n{p})  \wedge (I_2^+ \vee p) \wedge (I_2^- \vee \n{p}) \wedge ((\n{I_3^+} \wedge \n{p}) \vee (\n{I_3^-} \wedge p))& \implies \\ 
((I_1^+ \vee p)  \wedge (I_2^+ \vee p) \wedge \n{I_3^+} \wedge \n{p}) \vee ((I_1^- \vee \n{p})   
\wedge (I_2^- \vee \n{p}) \wedge  \n{I_3^-} \wedge p) & \implies^{\text{resol}}\\
(I_1^+  \wedge I_2^+ \wedge \n{I_3^+})  \vee  (I_1^-  \wedge I_2^- \wedge \n{I_3^-})& \implies^{\text{i.h.}}
\bot 
\end{flalign*} 
\end{adjustwidth}
\noindent{}All the remaining cases are treated in a similar manner, to
reach a point (possibly after a resolution step if some of the labels
are $ab$) where the inductive hypothesis can be applied.
\end{proof}

\begin{lemma}
If $\{Itp_{L_1},Itp_{L_2},Itp_{L_3}\}$ has BGSA, then $\{L_1,L_2,L_3\}$ satisfies $CC_{BGSA}$.
\end{lemma}
\begin{proof}[by contradiction]
We remind here the \emph{BGSA constraints} $CC_{BGSA}$: 
\[(\alpha_1,\alpha_2), (\delta_1,\delta_2) \preceq \{(ab,ab),(b,a),(a,b)\}, 
 \beta_2 \preceq \beta_3, \gamma_1 \preceq \gamma_3, \delta_1 \preceq \delta_3 ,\delta_2 \preceq \delta_3\]
 
 We  show that, if any of the $CC_{BGSA}$ constraints is violated,  
 there exist an unsatisfiable formula $\Phi = \{\phi_1, \phi_2, \phi_3\}$  and a refutation such that
 ${I_{\phi_1,L_1} \wedge I_{\phi_2,L_2}\centernot \implies I_{\phi_1 \phi_2, L_3}}$.
 The possible violations for the $CC_{BGSA}$ constraints consist of:
 \begin{enumerate}
 \item $(\alpha_1,\alpha_2),(\delta_1,\delta_2) \in \{(a,a),(ab,a),(a,ab)\}$
 \item $(\beta_2,\beta_3), (\gamma_1,\gamma_3), (\delta_1,\delta_3),(\delta_2,\delta_3) \in \{(a,ab),(a,b),(ab,b)\}$
\end{enumerate}
It is sufficient to take into account $(\alpha_1,\alpha_2) \in
\{(a,a),(a,ab)\}$ and $(\beta_2,\beta_3)\in
\{(a,ab),(a,b),(ab,b)\}$. The remaining cases follow by symmetry.

\begin{enumerate}
 \item[(1)] $(\alpha_1,\alpha_2)=(a,a): $ 
$\phi_1 = (p \vee \n{q}) \wedge r$, 
$\phi_2 = (\n{p} \vee \n{r}) \wedge  q$, $\phi_3 = s$\\

\begin{adjustwidth}{0em}{0em}
\begin{tabular}{c}
$A= \phi_1$ $\:$ $B= \phi_2 , \phi_3$\\ 
\begin{minipage}{0.4\textwidth}
\smallskip
\small
\begin{prooftree}
\def\defaultHypSeparation{\hskip.15in}
\AxiomC{$p \vee \n{q} \,[\bot]$} 
\AxiomC{$\n{p} \vee \n{r} \, [p \wedge r]$}
\BinaryInfC{$\n{q} \vee \n{r} \, [p \wedge r]$}
\AxiomC{$r\, [\bot]$}
\BinaryInfC{$\n{q} \, [p \wedge r]$}
\AxiomC{$q \, [\n{q}]$}
\BinaryInfC{$\bot \, [(p \wedge r)\vee \n{q}]$}
\end{prooftree}
\end{minipage}
\end{tabular}
\end{adjustwidth}
\begin{adjustwidth}{0em}{0em}
\begin{tabular}{c}
$ A= \phi_2$ $\:$ $B= \phi_1 , \phi_3$
\\
\begin{minipage}{0.4\textwidth}
\smallskip
\small
\begin{prooftree}
\def\defaultHypSeparation{\hskip.11in}
\AxiomC{$p \vee \n{q} \, [\n{p} \wedge q]$} 
\AxiomC{$\n{p} \vee \n{r} \,  [\bot]$}
\BinaryInfC{$\n{q} \vee \n{r} \, [\n{p} \wedge q]$}
\AxiomC{$r\, [\n{r}]$}
\BinaryInfC{$\n{q} \, [(\n{p} \wedge q) \vee \n{r}]$}
\AxiomC{$q \, [\bot]$}
\BinaryInfC{$\bot \, [(\n{p} \wedge q) \vee \n{r}]$}
\end{prooftree}
\end{minipage}
\end{tabular}
\\ \\ 
\end{adjustwidth}
We have $I_{\phi_1, L_1}= (p \wedge r)\vee \n{q}$, $I_{\phi_2, L_2}= (\n{p} \wedge q) \vee \n{r}$, 
$I_{\phi_1 \phi_2, L_3}= \bot$ since $s$ is absent from the proof.
Then, $I_{\phi_1,L_1} \wedge I_{\phi_2,L_2}\centernot \implies I_{\phi_1 \phi_2, L_3}$: 
a counter model is $\n{q}, \n{r}$.\\

 \item[(2)] $(\alpha_1,\alpha_2)=(a,ab)$ : $\phi_1 = (p \vee \n{q}) \wedge r $, 
$\phi_2 = (\n{p} \vee \n{r}) \wedge q$, $\phi_3 = s$\\

\begin{adjustwidth}{0em}{0em}
\begin{tabular}{c}
$ A= \phi_2$ $\:$ $B= \phi_1 , \phi_3$
\\
\begin{minipage}{0.4\textwidth}
\smallskip
\small
\begin{prooftree}
\def\defaultHypSeparation{\hskip.11in}
\AxiomC{$p \vee \n{q} \, [\top]$} 
\AxiomC{$\n{p} \vee \n{r} \,  [\bot]$}
\BinaryInfC{$\n{q} \vee \n{r} \, [\n{p}]$}
\AxiomC{$q\, [\bot]$}
\BinaryInfC{$\n{r} \, [(\n{p} \vee \n{q}) \wedge q]$}
\AxiomC{$r \, [\top]$}
\BinaryInfC{$\bot \, [((\n{p} \vee \n{q}) \wedge q) \vee \n{r}]$}
\end{prooftree}
\end{minipage}
\end{tabular}
\\ \\ 
\end{adjustwidth}
We have $I_{\phi_1, L_1}= (p \wedge r)\vee \n{q}$ and $I_{\phi_1 \phi_2, L_3}= \bot$ as in $(1)$, while
 $I_{\phi_2, L_2}= ((\n{p} \vee \n{q}) \wedge q) \vee \n{r}$.
 Then, $I_{\phi_1,L_1} \wedge I_{\phi_2,L_2}\centernot \implies I_{\phi_1 \phi_2, L_3}$: 
a counter model is $\n{q}, \n{r}$.\\
 \item[(3)] $(\beta_2,\beta_3)=(a,b)$ : $\phi_1 = s$, 
$\phi_2 = (\n{p} \vee \n{r}) \wedge  q$, $\phi_3 = (p \vee \n{q}) \wedge r$\\

\begin{adjustwidth}{0em}{0em}
\begin{tabular}{c}
$A= \phi_1, \phi_2$ $\:$ $B= \phi_3$ 
\\
\begin{minipage}{0.4\textwidth}
\begin{prooftree}
\smallskip
\small
\def\defaultHypSeparation{\hskip.11in}
\AxiomC{$p \vee \n{q} \, [\top]$} 
\AxiomC{$\n{p} \vee \n{r} \,  [\n{p} \vee \n{r}]$}
\BinaryInfC{$\n{q} \vee \n{r} \, [\n{p} \vee \n{r}]$}
\AxiomC{$r\, [\top]$}
\BinaryInfC{$\n{q} \, [\n{p} \vee \n{r}]$}
\AxiomC{$q \, [q]$}
\BinaryInfC{$\bot \, [(\n{p} \vee \n{r}) \wedge q]$}
\end{prooftree}
\end{minipage}
\end{tabular}
\\ \\ 
\end{adjustwidth}

We have $I_{\phi_1, L_1}= \top$, since $s$ is absent from the proof, while $I_{\phi_2, L_2}= (\n{p} \wedge q) \vee \n{r}$ as in $(1)$;
$I_{\phi_1 \phi_2, L_3}= (\n{p} \vee \n{r}) \wedge q$.
 Then, $I_{\phi_1,L_1} \wedge I_{\phi_2,L_2}\centernot \implies I_{\phi_1 \phi_2, L_3}$: 
a counter model is $\n{q}, \n{r}$.\\
 \item[(4)] $(\beta_2,\beta_3)=(a,ab)$ : $\phi_1 = s$, 
$\phi_2 = (\n{p} \vee \n{r})  \wedge  q $, $\phi_3 = (p \vee \n{q}) \wedge r $\\

\begin{adjustwidth}{0em}{0em}
\begin{tabular}{c}
$A= \phi_1 , \phi_2$ $\:$ $B= \phi_3$ \\
\begin{minipage}{0.4\textwidth}
\begin{prooftree}
\smallskip
\small
\def\defaultHypSeparation{\hskip.11in}
\AxiomC{$p \vee \n{q} \, [\top]$} 
\AxiomC{$\n{p} \vee \n{r} \,  [\bot]$}
\BinaryInfC{$\n{q} \vee \n{r} \, [\n{p}]$}
\AxiomC{$r\, [\top]$}
\BinaryInfC{$\n{q} \, [\n{p} \vee \n{r}]$}
\AxiomC{$q \, [\bot]$}
\BinaryInfC{$\bot \, [(\n{p} \vee \n{r} \vee \n{q}) \wedge q]$}
\end{prooftree}
\end{minipage}
\end{tabular}
\\ \\ 
\end{adjustwidth}

$I_{\phi_1, L_1}= \top$ as in $(3)$, $I_{\phi_2, L_2}= (\n{p} \wedge q) \vee \n{r}$ as in $(1)$,
$I_{\phi_1 \phi_2, L_3}= (\n{p} \vee \n{r} \vee \n{q}) \wedge q$.
 Then, $I_{\phi_1,L_1} \wedge I_{\phi_2,L_2}\centernot \implies I_{\phi_1 \phi_2, L_3}$: 
a counter model is $\n{q}, \n{r}$.\\
 \item[(5)] $(\beta_2,\beta_3)=(ab,b)$ : $\phi_1 = s$, 
$\phi_2 = (\n{p} \vee \n{r}) \wedge q$, $\phi_3 = (p \vee \n{q}) \wedge r$\\

$I_{\phi_1, L_1}= \top$ as in $(3)$, $I_{\phi_2, L_2}= ((\n{p} \vee \n{q}) \wedge q) \vee \n{r}$ as in $(2)$, 
$I_{\phi_1 \phi_2, L_3}= (\n{p} \vee \n{r}) \wedge q$ as in $(3)$.
 Then, $I_{\phi_1,L_1} \wedge I_{\phi_2,L_2}\centernot \implies I_{\phi_1 \phi_2, L_3}$: 
a counter model is $\n{q}, \n{r}$.
\end{enumerate}
\end{proof}

\begin{lemma}
\label{lem:ngsa_suff}
If $\{L_1,\ldots,L_{n+1}\}$ satisfies $CC_{nGSA}$, then  
the family $\{Itp_{L_1},\ldots,Itp_{L_{n+1}}\}$ has $n$-GSA. 
\end{lemma}
\begin{proof}[by structural induction]
We assume that the $CC_{nGSA}$ constraints have been restricted in a similar manner to what shown in $CC^*_{BGSA}$.
We prove that, given a refutation of $\Phi$, for any clause $C$ in the refutation 
the partial interpolants satisfy $I_{\phi_1,L_1}(C) \wedge \ldots \wedge I_{\phi_n,L_n}(C) \implies  I_{\phi_1 \ldots \phi_n,L_{n+1}}(C)$,
that is $I_{\phi_1,L_1}(C) \wedge \ldots \wedge  I_{\phi_n,L_n}(C) \wedge \n{I_{\phi_1 \ldots \phi_n,L_{n+1}}(C)} \implies  \bot$.

\bfstart{Base case} (leaf).
Remember that, if $C\in \phi_i, i\neq n+1$, $C$ has class $A$ in configuration $i$ 
(hence the partial interpolant is $C\!\!\downharpoonright_{i,b}$) and in configuration $n+1$ 
($\n{C\!\!\downharpoonright_{n+1,b}}$)
and class $B$ in all the other
configurations $j \neq i,n+1$ ($\n{C\!\!\downharpoonright_{j,a}}$).
If $C\in \phi_{n+1}$, it has class $B$ in all configurations 
($C\!\!\downharpoonright_{n+1,a}$ in configuration $n+1$,  $\n{C\!\!\downharpoonright_{i,a}}$ everywhere else).
So we need to prove: 
$$\n{C\!\!\downharpoonright_{1,a}} \wedge \ldots \wedge \n{C\!\!\downharpoonright_{i-1,a}} \wedge C\!\!\downharpoonright_{i,b} 
\wedge \n{C\!\!\downharpoonright_{i+1,a}} \wedge \ldots \wedge \n{C\!\!\downharpoonright_{n,a}} \wedge \n{C\!\!\downharpoonright_{n+1,b}}\implies \bot$$
$$\n{C\!\!\downharpoonright_{1,a}} \wedge \ldots \wedge \n{C\!\!\downharpoonright_{i-1,a}} \wedge \n{C\!\!\downharpoonright_{i,a}} 
\wedge \n{C\!\!\downharpoonright_{i+1,a}} \wedge \ldots \wedge \n{C\!\!\downharpoonright_{n,a}} \wedge C\!\!\downharpoonright_{n+1,a}\implies \bot$$
respectively for $i\neq n+1$ and $i = n+1$.

We can divide the variables of $C \in \phi_i$ into partitions, obtaining
$C = C_{\phi_i} \vee C_{\phi_i \phi_2} \vee \ldots \vee C_{\phi_1 \ldots \phi_n}$, leading to a system of constraints 
as shown for BGSA; the conjunction of:
$$\n{(C_{\phi_i} \vee  C_{\phi_i \phi_2} \vee \ldots \vee C_{\phi_1 \ldots \phi_n})}\!\!\downharpoonright_{1,a}$$
$$\vdots$$
$$(C_{\phi_i} \vee  C_{\phi_i \phi_2} \vee \ldots \vee C_{\phi_1 \ldots \phi_n})\!\!\downharpoonright_{i,b}$$
$$\vdots$$
$$\n{(C_{\phi_i} \vee  C_{\phi_i \phi_2} \vee \ldots \vee C_{\phi_1 \ldots \phi_n})}\!\!\downharpoonright_{n,a}$$
must imply $\bot$ for every $\phi_i,\: i\neq n+1$ (similarly for $\phi_{n+1})$.
All the simplifications are carried out in line with the proof of Lemma~\ref{lem:bgsa_suff}.

\bfstart{Inductive step} (inner node). The proof is a again a direct generalization of the proof of Lemma~\ref{lem:bgsa_suff}.

Performing a case splitting on the pivot and on its labeling vector, 
the starting point is a conjunction of the partial interpolants ${I_1 \wedge \ldots \wedge I_n \wedge  \n{I_{n+1}}}$ of $C$, 
which is then expressed in terms of the partial 
interpolants for the antecedents.
The goal is to reach a formula 
${\psi = (I_1^+ \wedge \ldots \wedge I_n^+ \wedge \n{I_{n+1}^+} )  \vee  (I_1^- \wedge \ldots \wedge I_n^- \wedge \n{I_{n+1}^-} )}$
where the inductive hypothesis can be applied. 

The key observation is that  the restricted $CC_{nGSA}$ constraints give rise to a combination of boolean operators 
(after the dualization of the ones in $\n{I_{n+1}}$ due to the negation) which makes it always possible to obtain the desired $\psi$,
possibly with the help of the resolution rule.
\end{proof}

\begin{lemma}
\label{lem:ngsa_nec}
If a family $\mF = \{Itp_{L_1},\ldots, Itp_{L_{n+1}}\}$ has $n$-GSA, then $\{L_1,\ldots,L_{n+1}\}$
satisfies $CC_{nGSA}$.
\end{lemma}
\begin{proof}[by induction and contradiction] 
  We prove the theorem by strong induction on $n\geq2$.

\bfstart{Base Case} ($n=2$). Follows by Lemma~\ref{lem:bgsa_nec}.

\bfstart{Inductive Step}. Assume the thesis holds for all ${k\leq
  n-1}$, we prove it for $k=n$.  By Lemma~\ref{lem:gsa_sub}, if a
family $\mF = \{Itp_{L_1},\ldots, Itp_{L_{n+1}}\}$ has $n$-GSA,
then any subfamily of size $k+1\leq n$ has $k$-GSA.  Combined with the
inductive hypothesis, this implies that it is sufficient to establish
the theorem for every variable $p$ and labeling vectors $\vec{\alpha}
= (\alpha_1, \ldots, \alpha_n)$ and $\vec{\beta} = (\beta_1, \ldots,
\beta_{n+1})$ corresponding to partitions $\phi_1 \cdots \phi_n$ and
$\phi_1 \cdots \phi_{n+1}$, respectively.

We only show the case of $\vec{\alpha}$. The proof for $\vec{\beta}$
is analogous. W.l.o.g., assume that there is a $p$ such that
$\vec{\alpha}$ violates $CC_{nGSA}$ for $\alpha_1 = \alpha_2 = a$
(other cases are symmetric). Construct a family of labelings
$\{L'_1,L'_2,L'_{n+1}\}$ from $\{L_1,\ldots,L_{n+1}\}$ by (1) taking
all labelings of partitions involving only subsets of $\phi_1$,
$\phi_2$ and $\phi_{n+1}$. For example, vectors $(\eta_3,\eta_4)$ and
$(\eta_1,\eta_2,\eta_3,\eta_{n+1})$ would be discarded, while
$(\eta_1,\eta_2)$ and $(\eta_1,\eta_2,\eta_{n+1})$ would be kept; and
(2) for $p$, set the labeling vector of partition $\phi_1\phi_2$ to
$(\alpha_1,\alpha_2) = (a,a)$. By Lemma~\ref{lem:bgsa_nec},
$\{L'_1,L'_2,L'_{n+1}\}$ does not have BGSA. Let $\Phi' =
\{\phi_1,\phi_2,\phi_{n+1}\}$ be such that
$I_{\phi_1,L'_1} \land I_{\phi_2,L'_2} \centernot \implies
I_{\phi_1\phi_2,L'_{n+1}}$, and let $\Pi$ be the corresponding resolution
refutation.

Construct $\Phi = \{\phi_1, \phi_2, p, \ldots, p, \phi_{n+1}\}$ by adding
$(n-2)$ copies of $p$ to $\Phi'$. $\Phi$ is unsatisfiable,
and $\Pi$ is also a valid refutation for $\Phi$. From this
point, we assume that all interpolants are generated from $\Pi$.

Assume, by contradiction, that $\mF$ has $n$-GSA. Then,
\[
I_{\phi_1,L_1} \land \cdots \land I_{\phi_n,L_n} \implies I_{\phi_1\cdots\phi_n,L_{n+1}}
\]
But, because $\phi_3, \ldots, \phi_n$ do not contribute any clauses to
$\Pi$, $I_{\phi_i,L_i} = \top$ for $3 \leq i \leq n$. Hence,
\[I_{\phi_1,L_1} \land I_{\phi_2,L_2} \implies I_{\phi_1\phi_2,L_{n+1}}\]
However, by construction:
\begin{align*}
  I_{\phi_1, L_1} &= I_{\phi_1,L'_1} & 
  I_{\phi_2, L_2} &= I_{\phi_2,L'_2} & 
  I_{\phi_1\phi_2, L_{n+1}} &= I_{\phi_1\phi_2,L'_{n+1}} 
\end{align*}
which leads to a contradiction. Hence $\vec{\alpha}$ must satisfy $CC_{nGSA}$.


\end{proof}

\begin{proposition}
\label{prop:2pi}
Any family  $\{Itp_{L_0},Itp_{L_1},Itp_{L_2}\}$ has $2$-PI.
\end{proposition}
\begin{proof}
  Recall that $I_{\top,L_0} = \top$ and $I_{\phi_1\phi_2,L_2} =
  \bot$ for any $L_0, L_2$. Hence, $2$-PI reduces to the following
  two conditions: 
    $\phi_1 \implies I_{\phi_1,L_1}$, $I_{\phi_1,L_1} \land \phi_2 \implies \bot$, 
  which are true of any Craig interpolant.
\end{proof}

\begin{corollary}
\label{cor:2sa}
A family  $\{Itp_{L_1},Itp_{L_2}\}$ has $2$-SA if and only if $\{L_1,L_2\}$ satisfies
$ (\alpha_1,\alpha_2) \preceq \{(ab,ab),(a,b),(b,a)\}$
\end{corollary}
\begin{proof}
 Follows from Lemma~\ref{lem:bgsa_nec} and Lemma~\ref{lem:bgsa_suff}.
\end{proof}

\begin{lemma}
  There exists a family $\{Itp_{L_0},Itp_{L_1},Itp_{L_2}\}$ that has
  $2$-PI and a family $\{Itp_{L'_1},Itp_{L'_2}\}$ that has $2$-SA, but the family 
  $\{Itp_{L_0},Itp_{L_1},Itp_{L_2},Itp_{L'_1},Itp_{L'_2}\}$ does not
  have $2$-STI.
\end{lemma}
\begin{proof} 
  By Theorem~\ref{lem:st-and-gsa}, a necessary condition for $2$-STI is
  that $\{Itp_{L_1},Itp_{L'_{2}},Itp_{L_{2}}\}$ has BGSA.  By
  Proposition~\ref{prop:2pi}, $\{L_0,L_1,L_2\}$ can be arbitrary.  By
  Theorem~\ref{theo:bgsa_nec_suff} and Corollary~\ref{cor:2sa}, there
  exists $\{L'_1, L'_2\}$ such that $\{Itp_{L'_1}, Itp_{L'_2}\}$ has $2$-SA, but
  $\{Itp_{L_1}, Itp_{L'_2}, Itp_{L_2}\}$ does not have BGSA.
\end{proof}

\begin{lemma}
The set of labeling constraints of any $n$-GSA strengthening is a subset of constraints of $n$-GSA.
\end{lemma}
\begin{proof}
Assume w.l.o.g we strengthen the first subformula $\phi_1$. Then any variable in 
any partition which does not involve $\phi_1$ has the same labeling vector and its
$n$-GSA labeling constraints are also the same. Instead, variables in any partition $\phi_1\phi_{i_2}\ldots\phi_{i_k}$
have now a labeling vector $(\alpha_{i_2},\ldots,\alpha_{i_k})$, where the first component $\alpha_1$ is missing.
Referring to the definition of $CC_{nGSA}$, it is easy to verify that the set of the constraints for the strengthening 
are a subset of the constraints for $n$-GSA. 
\end{proof}

\setcounter{theorem}{12}
\begin{theorem}
Any $\mF=\{Itp_{L_{i_1}},\ldots,Itp_{L_{i_k}},Itp_{L_{n+1}}\}$
s.t. $k< n$ that has an $n$-GSA strengthening property can be extended
to a family that has $n$-GSA.
\end{theorem}
\begin{proof}
  Refer to the definition of $CC_{nGSA}$ and to
  Lemma~\ref{lem:ngsa_str}. We can complete $\mF$ for example by
  introducing $n-k$ instances of McMillan's system $Itp_M$. Both
  constraints $(1)$ and $(2)$ for $n$-GSA are satisfied, since $Itp_M$
  always assigns label $b$ (recall the order $b\preceq ab \preceq
  a$). Note that $Itp_M$ is not necessarily the only possible choice.
\end{proof}

\begin{theorem}
PI holds for all single LISs.
\end{theorem}
\begin{proof}
  In \cite{RSS12} we addressed $n$-PI for a family of
  LISs $\{Itp_{L_0},\ldots,Itp_{L_n}\}$.  Given an inconsistent 
  $\Phi = \{\phi_1,\ldots,\phi_n\}$, Table~\ref{tab:npi} shows the
  labelings $L_i,L_{i+1}$ for an arbitrary step $I_{\phi_1\ldots
    \phi_i,L_i} \wedge \phi_{i+1} \implies I_{\phi_1\ldots \phi_i
    \phi_{i+1},L_{i+1}}$ ($\psi_1 = \phi_1 \wedge \ldots \wedge
  \phi_i$, $\psi_2 = \phi_{i+1}$, $\psi_3 = \phi_{i+2} \wedge \ldots
  \wedge \phi_n)$:

\begin{table}[h]%
\centering%
\caption{$n$-PI step.}%
\begin{tabular}{|l|c|c|}%
\hline 
\multicolumn{1}{|c|}{\multirow{2}{*}{$p$ in ?}} & \multicolumn{2}{c|}{Variable $class$, $label$}\\
\cline{2-3}
&$\psi_1 \mid \psi_2 \psi_3$   &   
$\psi_1  \psi_2 \mid \psi_3$\\
\hline
$\psi_1$ & $A,a$  & $A,a$\\
$\psi_2$ & $B,b$ & $A,a$\\
$\psi_3$ & $B,b$ & $B,b$\\
$\psi_1 \psi_2$ & $AB,\alpha_1$  & $A,a$\\
$\psi_2 \psi_3$ & $B,b$ &  $AB,\beta_2$\\
$\psi_1 \psi_3$ & $AB,\gamma_1$ &$AB,\gamma_2$\\
$\psi_1 \psi_2 \psi_3$ & $AB,\delta_1$ &$AB,\delta_2$\\
\hline%
\end{tabular}
\label{tab:npi}%
\end{table}%
We identified a set of constraints for $L_i,L_{i+1}$ as: 
\begin{align*}
\gamma_1 &\preceq \gamma_2 & \delta_1 &\preceq \delta_2
\end{align*}
For a single LIS, $\gamma_1 =\gamma_2$ and
$\delta_1=\delta_2$, so all constraints are trivially satisfied for
$0\leq i \leq n-1$.
\end{proof}

\section{Complexity of the Labeled Interpolation Systems}
\label{sec:app:3}

We briefly examine here the complexity of a Labeled Interpolation System $Itp_L$.
\begin{figure}[t]
\centering
\vspace{-1.3cm}
\begin{tabular}{|l c|c|c|c || l c|c|c|c|}%
\hline
Leaf: & \multicolumn{4}{c|}{$C \, [I]$} & Inner node: & \multicolumn{4}{c|}{$\quad \dfrac{C^+ \vee p:\alpha \, [I^+] \qquad  C^- \vee \n{p}:\beta \, [I^-]}{C^+ \vee C^- \, [I]}$} \\
\hline 
\multicolumn{5}{|c|}{$I =
\left\{
	\begin{array}{ll}
		C \!\!\downharpoonright b  & \quad \mbox{if } C \in A\\
		\neg (C \!\!\downharpoonright a) & \quad \mbox{if } C \in B\\
	\end{array}
\right. $
}
&
\multicolumn{5}{|c|}{$I =
\left\{
	\begin{array}{ll}
		I^+ \vee I^-  & \quad \mbox{if } \alpha \sqcup \beta = a\\
		I^+ \wedge I^- & \quad \mbox{if } \alpha \sqcup \beta = b\\
		(I^+ \vee p) \wedge (I^- \vee \n{p}) & \quad \mbox{if } \alpha \sqcup \beta = ab 
	\end{array}
\right. $
}\\
\hline 
\end{tabular}
\vspace{-0.14in}
\caption*{Labeled interpolation system $Itp_L$.}
\end{figure}
A simple realization of the interpolation algorithm of Fig.~\ref{tab:gen} (reported above) is based on a topological visit of the refutation DAG.

While visiting a leaf, the partial interpolant is computed by restricting the clause w.r.t. to $a$ or $b$, given a labeling $L$ for its shared variables.
Note that it is not necessary to specify labels for local variables, since variables of class $A$ can only have label $a$ and variables of class $B$
only label $b$.

While visiting an inner node,
(i) the labels of the shared variables of the resolvent clause are updated based on the labels of the antecedent clauses, (ii) the label of the pivot
is computed in the same way, and (iii) the partial interpolant is obtained by a boolean combination of the (already computed) partial interpolants of the antecedents,
plus possibly two occurrences of the pivot.

We distinguish between the complexity of generating partial interpolants for leaves and inner nodes as follows.

\bfstart{Leaf}. The cost of restricting a clause $C$ is $|C|$. Checking whether a clause or a variable has class $A,B,AB$ takes constant time. 

\bfstart{Inner node}. If $C$ is the resolvent clause, (i) takes $|C|$, both (ii) and (iii) take constant time.
We assume that, for each node, the labels of shared variables are encoded in a bit-vector-like data structure, so that retrieving the label
of a variable takes constant time.

Assume the DAG has $N$ nodes and the largest clause has size $S$, then the overall complexity is $O(NS)$.
In practice, $S << N$ and the complexity is linear in the size of the DAG. The overhead introduced by the computations due to the use
of a labeling is thus negligible.

%

\end{document}